\documentclass[preprint,12pt,nonatbib,fleqn]{elsarticle}
\usepackage{natbib}    % required for bibliography
\usepackage{graphicx,picture,calc}      % include this line if your document contains figures
\usepackage{lineno,hyperref}
\usepackage{xcolor}
\usepackage{amssymb,amsmath,amsthm}
\usepackage[separate-uncertainty=true]{siunitx}
\usepackage{epstopdf}
\usepackage{marvosym}
\usepackage{url}
\usepackage{enumitem}
\usepackage{multicol, blindtext}
\usepackage{subcaption}
\usepackage{import}
\usepackage{optidef}

\usepackage{caption}                % For subfigures
\usepackage{subcaption}             % For subfigures
\usepackage[capitalize]{cleveref}   % Clever references, automatically ads Fig etc

\usepackage{soul} % TODO delete

\graphicspath{{./Fig/}}

\newtheorem{theorem}{Theorem}
\newtheorem{proposition}{Proposition}
\newtheorem{corollary}{Corollary}
\newtheorem{remark}{Remark}

\newcommand*{\tran}{^{\mathrm{T}}}
\newcommand*{\nom}{\mathrm{nom}}
\newcommand*{\mean}{\mathrm{mean}}

\begin{document}
\begin{frontmatter}
\tnotetext[footnoteinfo]{%
This work was supported by the project 24-10301S of the Czech Science Foundation and by the European Union under the project Robotics and advanced industrial production No. ${\rm CZ}.02.01.01/00/22\_008/0004590$. The third author was supported by the project C14/22/092 of the Internal Funds KU Leuven and the project G092721N of the Research Foundation-Flanders (FWO - Vlaanderen).\\ \\
\textbf{Cite as:}\space Adam Peichl, Matěj Kuře, Wim Michiels, Tomáš Vyhlídal. \textit{Integrated design of system structure and delayed resonator towards efficient non-collocated vibration absorption.} Journal of Sound and Vibration, Volume 612, 2025, 119101. https://doi.org/10.1016/j.jsv.2025.119101 \\ \\
\copyright 2025 The Authors. Published by Elsevier Ltd. This is an open access article under the CC BY license (\href{https://creativecommons.org/licenses/by/4.0/}{https://creativecommons.org/licenses/by/4.0/}).\\ \\
\href{https://crossmark.crossref.org/dialog/?doi=10.1016/j.jsv.2025.119101&domain=pdf}{Check for updates}}

%\title{Optimizing structural properties towards non-collocated vibration absorption by delayed resonator - case study with experimental validation}
%\title{Simultaneous system structure optimization and delayed resonator design towards efficient non-collocated vibration absorption}
\title{Integrated design of system structure and delayed resonator towards efficient non-collocated vibration absorption}

\author[First]{Adam Peichl}
\author[First]{Matěj Kuře}
%\author[First]{Jaroslav Busek}
%\author[Second]{Nejat Olgac}
\author[Third]{Wim Michiels}
\author[First]{Tom\'{a}\v{s} Vyhl\'{\i}dal}

\address[First]{Department of Instrumentation
and Control Engineering, Faculty of Mechanical Engineering, Czech Technical University in Prague, Technick\'{a} 4, 166 07 Prague 6, Czechia,
{\tt}}
% \address[Second]{Department of Mechanical Engineering, University of Connecticut, Storrs, CT, USA
% {\tt  }}
\address[Third]{Department of Computer Science, KU Leuven, Celestijnenlaan 200A, B-3001 Heverlee, Belgium, {\tt}}

%\begin{abstract}
 %A thorough analysis is performed for a problem of non-collocated vibration absorption by a delayed resonator with position feedback, considering a system consisting of a series of flexibly linked single-degree-of-freedom masses being excited by a harmonic force. For the equilibrium stage when the target mass is fully stopped, key forces, motion amplitudes and potential energies across the structure are analysed. Next, a complete parameter set of the feedback gain and delayed pairs is derived. The analysis includes an assessment of the actuation force and power needed by the resonator for the full vibration absorption. The derived quantities are utilized in forming an optimization problem to balance the risk of fatigue across the system structure and power needed by the resonator, under the stability and parameter constraints. Next to the gain and delay of the DR feedback, selected structural parameters of the system form are used as variables in the constrained nonlinear optimization problem. The overall design procedure is demonstrated in a numerical case study on a system with five masses. 
%\end{abstract}
\begin{abstract}
The problem of non-collocated vibration absorption by a delayed resonator is addressed with emphasis on system  fatigue resistance and energy efficiency of control actions. The analysis is performed for a system consisting of an arbitrary large series of flexibly linked single-degree-of-freedom masses. For the stage where the vibration of the target mass is fully absorbed by the non-collocated resonator, key forces, motion amplitudes and potential energies across the system structure are assessed. Next, a complete parameter set of the resonator gain and delay is derived, and the actuation force and power needed by the resonator for the full vibration absorption is determined. The derived quantities are utilized in forming an optimization problem to balance minimal risk of fatigue across the system structure and power needed by the resonator, under the closed loop stability and parameter constraints. Next to the gain and delay of the resonator, selected structural parameters of the system are used as variables in the constrained nonlinear optimization problem. \textcolor{black}{Experimental and numerical case studies are included to demonstrate benefits of the proposed integrated structural and control design}.  
\end{abstract}

\end{frontmatter}

\section{Introduction}
Vibration absorbers have become established tools to mitigate adverse effects of vibration in a wide range of engineering applications. Traditionally, the vibration absorber is deployed at the place of the mechanical structure, where vibrations are to be suppressed. This setting will further be referred to as {\em collocated vibration absorption}. The collocated absorber tuning, by passive \cite{richiedeiTunedMassDamper2022}, semi-active \cite{yuan2019mode}, or active methods, \cite{preumontVibrationControlActive2011}, is a relatively straightforward task and has been widely addressed in literature.
%see e.g. the applications in milling \cite{yuan2019mode}, \cite{guo2016vibration}, surgery robotics \cite{sang2016fuzzy} and wind turbine \cite{zuo2017using}. 
In many applications, however, the absorber cannot be placed at the target location due to various constraints. 
As a typical example we refer to machining, where the absorber cannot be placed at the tip of the cutting tool. 
%or at the endpoint of micromanipulator of a surgery robot. 
In such applications, the vibration absorber needs to be deployed at a different position at the structure.
This setting will  be referred to in what follows as {\em non-collocated vibration absorption}. Thus, the vibration (and associated energy) need to propagate through a part of the system's structure before being absorbed. Inevitably, tuning of such an absorber is a considerably more difficult task since the dynamics of the involved part of the system need to be included in the design, \cite{jenkins_real-time_2019}, \cite{olgacActivelyTunedNoncollocated2021}. 

%\subsection{Methods related to non-collocated vibration absorption}
Due to the problem complexity, studies related to non-collocated vibration absorption of a harmonic excitation at frequency $\omega$ have been limited so far. In \cite{buhr1997non} a passive absorber is proposed to suppress vibrations at a remote location. 
In 
%\cite{mottershead_inverse_2001},
\cite{mottershead_inverse_2006},
the stiffness is passively adjusted by assigning zeros $\pm \jmath\omega$ to the {\em receptance}, i.e. to the transfer function between the excitation force and the target position to be damped, whereas 
state feedback is applied to achieve closed loop stability. 
In \cite{cha_imposing_2004} sprung masses are used to impose the points of zero vibration (so called nodes) for general elastic
structures during forced harmonic excitation. 
%It is shown for the problem at  hand that the non-collocated vibration suppression can be achieved only at certain locations along the structure. 
Maximum allowable amplitudes were studied in the subsequent work \cite{cha2005enforcing}.
%In \cite{cha2006imposing}, torsional vibration absorbers were applied to the translating ones to impose a zero slope at the nodes. 
By applying geometric methods, the authors of \cite{daley2008geometric}, \cite{wang2010broad}, derived an active control approach for simultaneous vibration attenuation at two separate locations in a system structure one of which is non-collocated. 
An electrical dynamic absorber is used in \cite{kim2013demonstration} to emulate the behavior of a non-collocated mechanical dynamic absorber on a setup with a flexible beam, see also \cite{kim2013modal} for an application of a modal filter.
A method for non-collocated vibration suppression in a light-weight structure has been proposed recently in \cite{richiedei_unit-rank_2022}. An active control with feedback from the non-collocated measurements is applied to shift the antiresonance frequency of a receptance to the excitation frequency $\omega$. %The method is experimentally validated by means of a cantilever-beam on a slip table. The experiments show a substantial reduction of oscillations at the non-collocated tip of the beam. Though, the vibration suppression is still far from being ideal. 
A pole-zero assignment to the receptance valid for the single-input and multi-input setting is proposed in 
\cite{richiedeiPolezeroAssignmentReceptance2022}. 
%Similarly, in \cite{xiangPartialPoleAssignment2016} and \cite{belottiPoleAssignmentVibrating2020}, the authors present methods of assigning poles applicable to a vibration suppression problem with time-delay in the feedback. 

%\subsection{Collocated vibration absorption by the Delayed Resonator}
In the above outlined works, the vibration suppression is achieved either by passive parameter adjustment, or by active feedback control of higher complexity. The method presented in this paper combines structural adjustment of the system and low-complexity control feedback at the active vibration absorber, analogous to feedback used in {\em Delayed Resonator} (DR). The DR concept was proposed in 1990s by Olgac and his co-workers \cite{olgac1994novel} for the standard collocated vibration absorption task. The time-delay feedback is applied to turn the active absorber into a marginally stable system with imaginary axis characteristic roots matching the excitation frequency, such that the DR acts as a perfect resonator and absorbs the vibration entirely. %A practical benefit of the DR is that in the standard collocated setting, neither measurements at the primary structure, nor its physical parameters are involved in the DR design. 
The resonator feedback can be implemented using position \cite{olgac1994novel}, velocity \cite{filipovic2002delayed}, or acceleration \cite{olgac1997active2} measurements at the absorber. Modifications of the resonator concept include a %relative position feedback in \cite{olgac1997active}, a 
torsional absorber \cite{hosek1997tunable}, and an auto-tuning algorithm to enhance the robustness against uncertainties \cite{hosek2002single}. 
%An algorithm for multi-degree of freedom (MDOF) mechanical structures is described in \cite{jalili1999multiple}, where multiple delayed resonators are used to suppress different harmonics. 
%The stability of the entire system is investigated using strategy so-called stability charts. 
%
%In \cite{rivaz2007active} a delay free PI alternative of the resonator has been proposed for the acceleration feedback. However, due to risky noise integration phenomenon, the feedback needs to be supplemented with high-pass filters, which makes the overall feedback more complex compared to time delay feedback. 
A complete dynamics analysis of the DR was performed in \cite{vyhlidal2014delayed} for the case of lumped delayed acceleration feedback, leading to a time delay system with neutral spectrum, and in \cite{7378505} for the case of distributed delayed  acceleration feedback, leading to a time delay system with retarded spectrum. The analysis was further extended in \cite{vyhlidal2019analysis} to DRs with position and velocity feedback, where the system's spectrum is also retarded. 
From the stability analysis investigated in \cite{vyhlidal2019analysis}, 
%using the {\em Cluster Treatment of Characteristic Roots} (CTCR) \cite{olgac2005cluster}, see also stability analysis using Nyquist criterion in \cite{alujevic2012tuneable}, 
it was demonstrated that if not tuned properly, the DR can considerably lower the stability margin or  destabilize the overall system even for its collocated deployment. 
%Note that in the classical DR setting, the two parameters (gain and delay) are used to assign a single pole couple $\pm \mathrm{j}\omega$, which is turned to an active zero of the receptance. Thus, no degree of freedom is left to regain the stability if it is lost due to DR deployment.
A way to handle this issue consists of the application of a multi-parameter DR feedback as proposed in \cite{PilbauerRobust}, where the stability constraint is included in the optimal design of a robust DR. Alternatively, an additional control loop can be included to regain and robustify the stability, as proposed in  \cite{kure2024robust}, while also increasing robustness with respect to frequency mismatch by placing a double root at the nominal frequency.
%It was shown that the operable frequency range of the DR is limited regardless the feedback type. From below, it is limited by the stability boundary, while the delay implementation aspects limit the range from the above due to the exponential decay of the delay length with respect to growing frequency. 
%A methodology for extending the operable frequency range was proposed in \cite{kuvcera2017extended}. 
%It is based on extending the feedback by a delay-free factor virtually modifying the mass of the absorber and thus its natural frequency.  
Other recent research directions include combining position and velocity feedback \cite{oytun2018new}, optimizing the resonator with multiple distributed delays \cite{liu2024delayed}, targeting two frequencies \cite{valasek2019real}, the deisgn of fractional-order DRs \cite{cai2023spectrum}, supplementing the DR with an amplifying mechanism \cite{liu2023robust}, and the dimensional extension of the concept to 2D \cite{vyhlidalAnalysisOptimizedDesign2022}, \cite{vsika2021two}, and 3D \cite{vsika2024three}, \cite{benevs2024collocated}. 
%A fractional order DR was proposed and analysed in \cite{cai2023spectrum}. 
%The work \cite{liu2023robust} supplements the DR with an amplifying mechanism.  
%
%Before moving further, let us point to alternative, delay-free, feedback types applied to turn the absorber to a resonator. 
In \cite{rivaz2007active} a delay free, but filtered, proportional-integral acceleration feedback of the absorber was elaborated.
%, supplemented by a low-pass filter to reduce the effect of noise and a high-pass filter to remove the effect of {\em random walk phenomenon} caused by noise integration. 
%The filter application makes the overall feedback more complex compared to time-delay feedback. 
In \cite{filipovic1999vibration}, the concept of {\em linear active resonator} (LAR) was proposed by Filipovic and Schr\"{o}der. Conceptually, it mirrors the DR structure with a tuneable gain, which, however, is placed in  series with a rational transfer function instead of the delayed term used in DR. %Next to design of various transfer function types aiming either at single or multiple frequencies, a stability analysis is performed for both the continuous and discrete time setting in \cite{filipovic1999vibration}.  

%\subsection{Non-collocated vibration absorption by the resonator absorber}
Recently, Olgac and Jenkins demonstrated that the DR is applicable for non-collocated vibration absorption of a system composed of a serial interconnection of flexibly linked masses \cite{jenkins_real-time_2019}, \cite{olgacActivelyTunedNoncollocated2021}, as shown in Fig. \ref{fig:non-collocated-scheme}. However, compared to the collocated DR design, part of the primary structure needs to be included in tuning the two parameters of the DR feedback. The DR together with this part of the primary form a {\em resonant substructure}, which needs be tuned as a whole. Thus, the pole couple $\pm \jmath\omega$, which becomes an active zero of the receptance, is to be assigned to the whole resonant substructure. Although such a pole placement is a relatively straightforward task, due to an increase of the structural complexity, the risk of stability loss of the whole system is even higher compared to the collocated case. It was also discussed in \cite{jenkins_real-time_2019}, \cite{olgacActivelyTunedNoncollocated2021} that the resonant substructure can only be defined in special system configurations, including the nominally analysed series of the masses. Note that the findings of Olgac and Jenkins confirm the earlier results by Filipovic and Schr\"{o}der presented in \cite{filipovic2001control}, where an analogous problem of remote (non-collocated) vibration suppression at a system composed of a series of  flexibly linked masses is solved by the LAR. Next to the theoretical analysis performed over the receptance transfer functions, an experimental validation is performed for the two-mass setup.   

In %\cite{silmStabilizationZeroLocation}
\cite{silm_2024_spectral_design}, a generally applicable spectral design of non-collocated vibration
suppression performed primarily by a DR is presented.
The vibration suppression is achieved by direct assignment of the receptance zeros to $\pm \jmath\omega$. In order to increase the stability margin, a stabilizing controller can be included and tuned by spectral optimization. 
Compared to the method of \cite{jenkins_real-time_2019}, \cite{olgacActivelyTunedNoncollocated2021}, the method proposed in %\cite{silmStabilizationZeroLocation}
\cite{silm_2024_spectral_design} does not require the existence of a resonating substructure. On the other hand, a part where absorbed energy is accumulated cannot be identified for this general setting as a rule. 
%Thus, the method is at the edge between general vibration suppression and vibration absorption. 
The proposed method has been validated on an experimental setup with three masses and a resonator actuated by voice-coils. In a subsequent work \cite{Saldanha2023}, an output feedback controller is used to assign the active zero couple and to stabilize the system with feedback delay. The synthesis is performed by spectral optimization and the results are also validated experimentally. In \cite{saldanha2023IFAC} the simultaneous zero assignment and stabilization in the non-collocated setting is achieved by a DR with multiple delay static feedback. A DR with delayed dynamic-feedback controller is proposed in \cite{saldanha2024TDS} to target multi-frequency non-collocated vibration absorption.

%\subsection{Problem formulation and paper outline}
Despite the above discussed development of generally applicable methods %\cite{silmStabilizationZeroLocation}, 
\cite{silm_2024_spectral_design},
\cite{Saldanha2023}, \cite{saldanha2023IFAC}, \cite{saldanha2024TDS}, the serial configuration addressed in \cite{jenkins_real-time_2019},  \cite{olgacActivelyTunedNoncollocated2021} and \cite{filipovic2001control} deserves further attention and deeper analysis as it is the configuration often met in practice. In \cite{jenkins_real-time_2019}, \cite{olgacActivelyTunedNoncollocated2021}, it was shown 
%by applying the CTCR method, \cite{olgac2005cluster}, 
that for given parameters of the system and the DR, only frequencies that correspond to stability regions in the parameter set can be successfully suppressed. 
Another important observation and a problem opened in \cite{jenkins_real-time_2019} and \cite{olgacActivelyTunedNoncollocated2021} concern fatigue loading across the structure. It was demonstrated by simulations, that, as a consequence of suppressing vibrations at the  non-collocated target mass, the amplitudes of the steady-state oscillations across the structure either increase or decrease compared to the passive case. Therefore, it imposes either higher or lesser fatigue loading. 
\color{black}
%Consequently, the paper \cite{olgacActivelyTunedNoncollocated2021} is concluded by the open problem question: {\em How does one design a DR to create reduced fatigue effects while performing real-time tunable and complete vibration suppression at a desired remote position?} By detailed analysis performed in the paper, we provide an answer to this question: 

The main objective of this paper is to address the two potentially risky features of non-collocated vibration absorption
identified by Olgac and Jenkins in the design methodology, \cite{jenkins_real-time_2019}, \cite{olgacActivelyTunedNoncollocated2021}, i.e., increased fatigue loading and impact on stability. Additionally, we aim at minimization of energy needs by the DR in the non-collocated vibration absorption task. 
Through detailed analysis of system model in phasor form, in \cref{section:structural-analysis-and-DR-design}, we show that for given parameters of the primary, the fatigue loading across the structure can be changed neither by choosing an alternative control method nor by redesigning the absorber parameters. It can only be changed by modifying the parameters of the structure. Such parameter adjustment can also be beneficial towards either gaining or strengthening stability of the structure with non-collocated DR. Next, in terms of decomposed phasor model characteristics, we provide the DR feedback synthesis and assessment of energy needed for the non-collocated vibration absorption. 
Utilizing these characteristics, an integrated design procedure is proposed in \cref{section:composed-structural-optimization} for the simultaneous optimization of the system's structural properties and the DR's feedback loop. 
For a given excitation frequency $\omega$, the objectives are set as: 
i) suppress fully vibration at the target non-collocated location with a sufficient stability margin,
ii) minimize fatigue loading 
     at the most risky place at the system structure, and iii)
minimize the power needed for vibration suppression by the active feedback.
 The objective i) is turned to constraints, whereas the objectives ii) and iii) are balanced within the objective function. This leads to a constrained nonlinear optimization problem to be solved. Two case studies are included to demonstrate and validate the proposed results. In \cref{section:case-study-experimental}, the mechanisms by which the optimized parameters influence the optimization objectives is demonstrated and experimentally validated on a mechatronic set-up with three masses. Consequently, for a more complex system structure with five masses, a thorough numerical case study is performed in \cref{section:case-study}. It is shown that the constrained nonlinear optimization problem can in principle be solved by recently proposed GRANSO tool \cite{curtis2017bfgs}, utilizing  functionalities of the TDS-CONTROL toolbox \cite{appeltans2022tds}.
%
%The rest of the paper is organized as follows. In \cref{section:structural-analysis-and-DR-design}, the structural analysis is performed for series of flexibly linked masses in the steady-state oscillation phase of the non-collocated vibration absorption. Next to that, design formulas for DR with position feedback are derived, which is followed by stability analysis. In \cref{section:composed-structural-optimization}, the simultaneous optimization of the system structural properties and DR design is formulated. 
%The application of the proposed methods is demonstrated in the case study in \cref{section:case-study}, 
Conclusions and future research directions are stated in \cref{section:conclusions}. 
\color{black}
\section{System structural analysis and delayed resonator design}
\label{section:structural-analysis-and-DR-design}
 % \textcolor{red}{Adam: Scheme with description and general model - defining the matrices, variables, modal equations and vector form.... }
% \begin{figure}
%     \centering
%     % \includegraphics[width=0.5\columnwidth]{Noncollocated_uvolneni_3.eps}
%     % Original size: 468x488pt
%     \includegraphics[width=\textwidth]{"./NCss_Adam_optimization_general.pdf"}
%     \caption{Top - overall setup of linked masses with mounted absorber $m_a$; middle - absorber is separated from the linear chain of masses and link is replaced with $f_a$ - total force acting on absorber; bottom - final stage is reached and mass $m_s$ is perfectly silenced, system is split into resonating (green) and vibrating (blue) subsystems.}
%     \label{fig:non-collocated-scheme}
% \end{figure}
\begin{figure}[t]
    \centering
    \begin{subfigure}[t]{\textwidth}
        \centering
        % Original size: 468x160pt
        \includegraphics[width=\textwidth]{"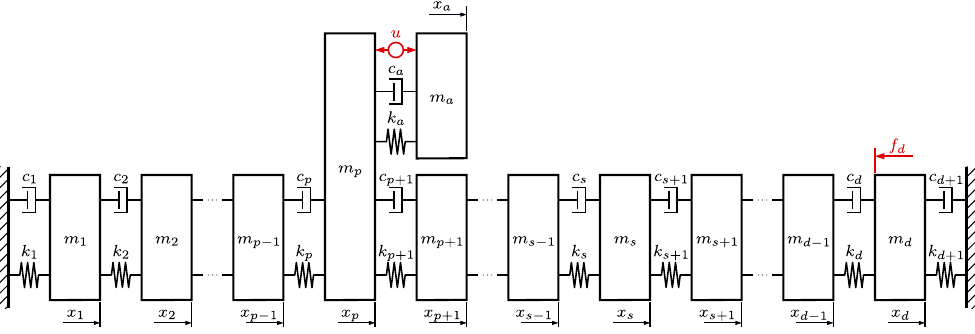"}
        \caption{The overall system of flexibly linked masses being excited by an $\omega$-harmonic force $f_d$, with an active absorber $m_a$ mounted at mass $m_p$, and the objective to ideally suppress vibration at the mass $m_s$}
         \label{subfig:non-collocated-scheme-overall}
    \end{subfigure}
   % \begin{subfigure}[t]{\textwidth}
   %     \centering
        % Original size: 468x4160pt
   %     \includegraphics[width=\textwidth]{"./NCss_Adam_optimization_general__separated.pdf"}
   %     \caption{Substitution of the effect of absorber by $f_a$ denoting {\em total force in the absorber's link}}
    %     \label{subfig:non-collocated-scheme-separated-absorber}
    %\end{subfigure}
    \begin{subfigure}[t]{\textwidth}
        \centering
        % Original size: 468x160pt
        \includegraphics[width=\textwidth]{"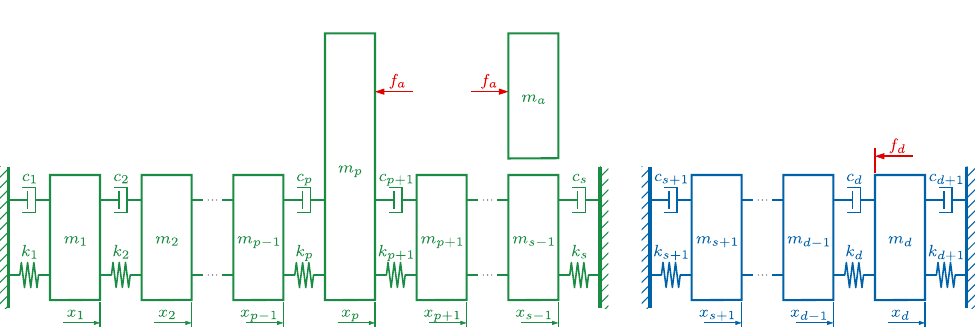"}
        \caption{Separation of the system when the vibration at the mass $m_s$ is ideally suppressed (i.e., it can be considered as a motion-less barrier) into i) a {\em resonating subsystem} (green) excited by $f_a$, ii) {\em vibrating subsystem} (blue) excited by $f_d$, and iii) {\em absorber} being excited by $f_a$.}
         \label{subfig:non-collocated-scheme-final-stage}
    \end{subfigure}
    % Original size: 468x488pt
    % \includegraphics[width=\textwidth]{"./NCss_Adam_optimization_general.pdf"}
    \caption{System configuration and its structural decomposition.}
    \label{fig:non-collocated-scheme}
\end{figure}

 Assume a linear chain of $d$ flexibly linked masses $m_i$ as shown in Fig. \ref{fig:non-collocated-scheme} a). Each mass $m_i$ with position $x_i$ is connected to its predecessor via a spring $k_i$ and a dumper $c_i$. The first and the last masses ($m_1$ and $m_d$) are connected to a base. Assume the $\omega$-harmonic disturbance force 
  \begin{equation}\label{eq:disturbance-cosine-form}
     f_d(t) = \bar{f}_d \cos{\omega t}
 \end{equation}
with amplitude $\bar f_d$ acts on mass $m_d$. The objective is to fully suppress vibration of the mass $m_s$. It is to be done by an active \emph{vibration absorber} with mass $m_a$, position $x_a$, damping $c_a$, and stiffness $k_a$, and an active control force $u(t)$. The absorber is deployed at the mass $m_p$. In the analysis which follows, we assume $p < s \leq d$, implying solution of \emph{non-collocated vibration absorption} problem. The case, where $p=s$ degrades to the standard problem of {\em collocated vibration absorption}, thoroughly analyzed, e.g., in \cite{vyhlidal2019analysis}. Note also that for $p < s = d$, i.e., the periodic force is acting directly on mass $m_s$, which substantially simplifies the problem to be solved, as demonstrated below. 

 For analytical purposes, let us introduce a \emph{total force in the absorber's link} as
\begin{equation}\label{eq:force}
   f_a(t)=u(t)+k_a(x_p(t)-x_a(t))+c_a(\dot{x}_p(t)-\dot{x}_a(t)), 
 \end{equation}
which allows us to free the absorber mass $m_a$ from the mass $m_p$, as shown in  Fig. \ref{fig:non-collocated-scheme} b). 
The system of linked masses $m_i,\ i=1,\dots,d$  can then be described by a set of $d$ second-order linear equations
 \begin{equation}\label{eq:system-general-MCK}
   \mathbf{M}\ddot{\mathbf{x}}(t) + \mathbf{C}\dot{\mathbf{x}}(t) + \mathbf{K}\mathbf{x}(t) = \mathbf{B}_d f_d(t) + \mathbf{B}_a f_a(t),
 \end{equation}
 where $\mathbf{x}(t)$ is vector of displacements $x_i,\ i=1, \dots, d$, \emph{mass matrix} $\mathbf{M}$ is of a diagonal form
 \begin{equation}
     \mathbf{M} = \text{diag}\left(m_1, m_2, \dots, m_p, \dots, m_s, \dots, m_{d-1}, m_d\right).
 \end{equation}
 Defining $\mathbf{k} = \left[k_2, \dots, k_p, \dots, k_s, \dots, k_{d-1}, k_d\right]$, and denoting zero column vector of length $n$ as $\mathbf{o}_{n}$, the \emph{stiffness matrix} has the tridiagonal structure
 \begin{equation}
     \mathbf{K}
     =\text{diag}([ k_1 \ \mathbf{k}]) + \text{diag}([ \mathbf{k} \ k_{d+1}])-
     \begin{bmatrix}
         \mathbf{o}_{d-1} & \text{diag}(\mathbf{k})\\
           0    & \mathbf{o}_{d-1}\tran
     \end{bmatrix}
     -\begin{bmatrix}
         \mathbf{o}_{d-1}\tran & 0\\
           \text{diag}(\mathbf{k})    & \mathbf{o}_{d-1}
     \end{bmatrix},
 \end{equation}
%  %%%%%%%%%%%%%%%%%%%%%
%  %Old matrix
%  %%%%%%%%%%%%%%%%%%%%
%  \color{red}
%  Left for check
%  \begin{equation}
%      \mathbf{K}
%      =
%      \begin{bmatrix}
%          k_1 + k_2 & -k_2 & & & & & & 0 \\
%          -k_2      & k_2 + k_3 & -k_3 & & & \\
%          & \ddots & \ddots & \ddots \\ 
%         & & -k_{p} & k_{p} + k_{p+1} & -k_{p+1} & & \\
%         & & & \ddots & \ddots & \ddots \\
%         & & & & -k_{s} & k_{s} + k_{s+1} & k_{s+1}\\
%         & & & & & \ddots & \ddots & \ddots \\
%         & & & & & -k_{d-1} & k_{d-1} & - k_{d} \\
%         0 & & & & & & -k_{d} & k_{d} + k_{d+1} \\
%      \end{bmatrix},
%  \end{equation}
%  \color{black}
% %%%%%%%%%%%%%%%%%%%%%%%%%%%%
 and \emph{damping matrix} $\mathbf{C}$ is of the same structure as $\mathbf{K}$, provided that each spring stiffness $k_i$ is replaced by the corresponding damping $c_i$.  
  %The column vectors $\mathbf{B}_d$ and $\mathbf{B}_a$ are in general determined by the location of the disturbance force and the absorber, respectively. In our case, the disturbance is acting on the last mass $m_d$ and the absorber is acting on the mass $m_p$. 
  The system input matrices are given by
 \begin{equation}\label{eq:system-general-Bd}
    \mathbf{B}_d = \begin{bmatrix} \mathbf{o}_{d-1}\tran & 1 \end{bmatrix}\tran,
 \end{equation}
 and
 \begin{equation}\label{eq:system-general-Ba}
    \mathbf{B}_a = \begin{bmatrix} \mathbf{o}_{p-1}\tran 
& 1 & \mathbf{o}_{d-p}\tran \end{bmatrix}\tran.
 \end{equation}

%\subsection{System decomposition in phasor form}
 In order to study the steady-state harmonic motion of the masses at the frequency $\omega$, let us turn the harmonic force \eqref{eq:disturbance-cosine-form} to
 \begin{equation}\label{eq:system-general-vector-disturbance}
    f_d(t) = \Re\{ \vec{f}_d e^{\jmath\omega t}\},
 \end{equation}
 with $\vec{f_d}$ being a phasor. The overall system \eqref{eq:system-general-MCK} can then be transformed to
 \begin{equation}\label{eq:sys-general-vector-form}
     \mathbf{A}(\omega) \vec{\mathbf{x}} = \mathbf{B}_d\vec{f}_d + \mathbf{B}_a\vec{f}_a,
 \end{equation}
 where the \emph{dynamic stiffness characteristic matrix} is given by
 \begin{equation}\label{eq:system-general-DSM}
     \mathbf{A}(\omega) = -\mathbf{M}\omega^2 + \mathbf{C} \jmath\omega + \mathbf{K}.
 \end{equation}
 From \eqref{eq:sys-general-vector-form} the phasor representing the displacement of mass $m_i$ can be expressed as
\begin{equation}\label{eq:system-displacement-mass-i-vector-form}
    \vec{x}_i = \mathbf{e}_i\tran \mathbf{A}^{-1}(\omega) \left( \mathbf{B}_d\vec{f}_d + \mathbf{B}_a\vec{f}_a \right),
\end{equation}
where $\mathbf{e}_i=\begin{bmatrix}  \mathbf{o}\tran_{i-1} & 1 & \mathbf{o}\tran_{d-i}\end{bmatrix}\tran$ encodes the position of mass $m_i$ with respect to the whole chain of masses.

Considering the objective of stopping mass $m_s$, with a reference to \cref{fig:non-collocated-scheme}, the system can be split into three subsystems: i) the {\em resonating subsystem} given by the masses $m_j,\ j=1, 2, \dots, s-1$; ii) the {\em target mass} $m_s$ to be stopped; and iii) the {\em vibrating subsystem} given by masses $m_k,\ k=s+1, s+2, \dots, d$. Respecting this subsystem division, the vector of displacements phasor is formed as
 \begin{equation}\label{eq:subsystems-positions-vector-form}
     \vec{\mathbf{x}} = \begin{bmatrix} \vec{\mathbf{x}}_R\tran & \vec{x}_s & \vec{\mathbf{x}}_V\tran \end{bmatrix}\tran,
 \end{equation}
 where $\vec{\mathbf{x}}_R=
\begin{bmatrix} \vec{x}_1 & \dots & \vec{x}_{s-1} \end{bmatrix}\tran$ and $\vec{\mathbf{x}}_V=\begin{bmatrix} \vec{x}_{s+1} & \dots & \vec{x}_{d} \end{bmatrix}\tran$ are vectors of displacements of resonating and vibrating subsystems, respectively. Analogously, we can express the dynamic stiffness matrix as
 \begin{align}\label{eq:subsystems-DSM}
     \mathbf{A}(\omega)
     =
     \begin{bmatrix}
      \mathbf{A}_R(\omega) & \mathbf{a}_{R}(\omega) & \mathbf{0} \\
      \mathbf{a}\tran_{R}(\omega) & a_{s,s}(\omega) & \mathbf{a}\tran_{V}(\omega) \\
      \mathbf{0}  & \mathbf{a}_{V}(\omega) & \mathbf{A}_V(\omega)
    \end{bmatrix},
 \end{align}
 where $\mathbf{a}_{R}\tran = \begin{bmatrix} \mathbf{o}_{s-1}\tran & - \jmath\omega c_{s} - k_{s}\end{bmatrix}$ and  $\mathbf{a}_{V}\tran = \begin{bmatrix} - \jmath\omega c_{s+1} - k_{s+1} & \mathbf{o}_{d-s}\tran\end{bmatrix}$ represent links from vibrating and resonating subsystem to target mass $m_s$, respectively. Input matrices are defined as
 \begin{align}
     \mathbf{B}_d = \begin{bmatrix} \mathbf{o}_{s-1}\tran & 0 & \mathbf{D}_d\tran  \end{bmatrix}\tran, \label{eq:subsystems-Bd} \\
     \mathbf{B}_a = \begin{bmatrix} \mathbf{D}_a\tran & 0 &\mathbf{o}_{d-s}\tran \end{bmatrix}\tran, \label{eq:subsystems-Ba}
 \end{align}
 where $\mathbf{D}_d=\begin{bmatrix}\mathbf{o}_{d-s-1}\tran & 1\end{bmatrix}\tran$ and $\mathbf{D}_a=\begin{bmatrix}\mathbf{o}_{p-1}\tran & 1 & \mathbf{o}_{s-p}\tran\end{bmatrix}\tran$.

\subsection{Force and motion balance at the target stage}
 Assume the \emph{target stage} with  
 \begin{align}\label{eq:assume-stop-ms}
     \vec{x}_s = 0
 \end{align}
 is reached thanks to the active absorber action fully compensating the effect of excitation force \eqref{eq:system-general-vector-disturbance}. Then, as indicated in Fig. \ref{subfig:non-collocated-scheme-final-stage}, the mass $m_s$ can be considered as a motion-less barrier and 
 the system can be split into {\em vibrating subsystem}  
 \begin{equation}\label{eq:sys-vibrating-vector-form}
     \mathbf{A}_V(\omega) \vec{\mathbf{x}}_V = \mathbf{D}_d\vec{f}_d,
 \end{equation}
and {\em resonating subsystem} 
  \begin{equation}\label{eq:sys-resonating-vector-form}
     \mathbf{A}_R(\omega) \vec{\mathbf{x}}_R = \mathbf{D}_a\vec{f}_a.
 \end{equation}
 The relation between the required force $\vec{f}_a$ and disturbance force $\vec{f}_d$ is clarified in the following proposition.
  
\begin{proposition}
%[Total force in the absorber's link]
    \label{proposition:total-absorber-force}
    Assume the system \eqref{eq:system-general-MCK} turned into the phasor form \eqref{eq:sys-general-vector-form}, with the output \eqref{eq:subsystems-positions-vector-form}, and matrices \eqref{eq:subsystems-DSM}, \eqref{eq:subsystems-Bd} and \eqref{eq:subsystems-Ba}, is excited by a harmonic force $f_d$ with frequency $\omega$ given in \eqref{eq:system-general-vector-disturbance}. Furthermore, assume $p < s \leq d$, $a_{s,s} \neq 0$ and invertibility of submatrices $\mathbf{A}_R(\omega)$ and $\mathbf{A}_V(\omega)$. Then, as soon as the motion of the {\em target mass} $m_s$ is perfectly stopped (i.e. $\vec{x}_s = 0$) the phasor of the total force in the absorber's link is given by
    \begin{equation}\label{eq:prop1-fa-normal-case}
     \vec{f_a} = -\left(\mathbf{a}_V\tran(\omega) \mathbf{A}_V^{-1}(\omega) \mathbf{D}_d \right) \left(\mathbf{a}_R\tran(\omega
) \mathbf{A}_R^{-1}(\omega) \mathbf{D}_a\right)^{-1} \vec{f}_d,
    \end{equation}
    for $p<s<d$ and by
    \begin{equation}\label{eq:prop1-fa-marginal-case}
     \vec{f}_a = \left(\mathbf{a}_R\tran(\omega) \mathbf{A}_R^{-1}(\omega) \mathbf{D}_a \right)^{-1} \vec{f}_d,
    \end{equation}
    for $p<s=d$.
\end{proposition}
\begin{proof}
    For $p<s<d$, we rewrite the system \eqref{eq:sys-general-vector-form} into
    \begin{align}
        \mathbf{A}_R(\omega) \vec{\mathbf{x}}_R + \mathbf{a}_R(\omega) \vec{x}_s = \mathbf{D}_a \vec{f}_a,\label{eq:prop1-proof-resonating} \\
        \mathbf{a}_R\tran(\omega) \vec{\mathbf{x}}_R + a_{s,s}(\omega) \vec{x}_s + \mathbf{a}_V\tran(\omega) \vec{\mathbf{x}}_V = 0, \label{eq:prop1-proof-stopped}\\
        \mathbf{a}_R(\omega) \vec{x}_s + \mathbf{A}_V(\omega) \vec{\mathbf{x}}_V = \mathbf{D}_d \vec{f}_d. \label{eq:prop1-proof-vibrating}
    \end{align}
    Substituting $\vec{x}_s = 0$ to \eqref{eq:prop1-proof-resonating}, \eqref{eq:prop1-proof-stopped} and \eqref{eq:prop1-proof-vibrating}, and solving \eqref{eq:prop1-proof-resonating} for $\vec{\mathbf{x}}_R$, and \eqref{eq:prop1-proof-vibrating} for $\vec{\mathbf{x}}_V$, respectively, yields
    \begin{equation}\label{eq:prop1-proof-solved-xr}
        \vec{\mathbf{x}}_R = \mathbf{A}_R^{-1}(\omega) \mathbf{D}_a \vec{f}_a,
   \end{equation}
   and
   \begin{equation}\label{eq:prop1-proof-solved-xv}
        \vec{\mathbf{x}}_V = \mathbf{A}_V^{-1}(\omega) \mathbf{D}_d \vec{f}_d.
    \end{equation}
    Substituting \eqref{eq:prop1-proof-solved-xr} and \eqref{eq:prop1-proof-solved-xv} into \eqref{eq:prop1-proof-stopped} and expressing $\vec{f}_a$ gives \eqref{eq:prop1-fa-normal-case}.

    For the case $p<s=d$, system of equations degenerates to
    \begin{align}
        \mathbf{A}_R(\omega) \vec{\mathbf{x}}_R + \mathbf{a}_R(\omega) \vec{x}_s = \mathbf{D}_a \vec{f}_a,\label{eq:prop1-proof-resonating-degraded} \\
        a_{s,s}(\omega) \vec{x}_s + \mathbf{a}_R\tran(\omega) \vec{\mathbf{x}}_R = \vec{f}_d.\label{eq:prop1-proof-stopped-vibrating-degraded}
    \end{align}
    Analogously as in the above case, substituting $\vec{x}_s = 0$ to \eqref{eq:prop1-proof-resonating-degraded} and \eqref{eq:prop1-proof-stopped-vibrating-degraded} and eliminating $\vec{\mathbf{x}}_R$ from the equation set, yields \eqref{eq:prop1-fa-marginal-case}.
 \end{proof}
 %Note that the total force $\vec{f}_a$ to be generated by the absorber depends only on the structure and the disturbance force $\vec{f}_d$.

\subsection{Fatigue loading expressed in terms of elastic potential energy at the links}
In general, stopping the mass $m_s$ leads to a redistribution of the oscillation amplitudes among the remaining masses. In some cases, it can go above safe limit and even cause damage due to accumulated fatigue. In order to assess this risk, let us express the elastic potential energy in $i^{\text{th}}$-link as 
\begin{equation}\label{eq:elastic-energies-general-timedomain}
     W_i(t) = \frac{1}{2}k_i\left( x_i(t) - x_{i-1}(t) \right)^2,
\end{equation}
 which can be considered as a measure to evaluate the risk of fatigue.
 
 Let $\Delta x_{i,i-1}(t)=x_i(t)-x_{i-1}(t)$ 
 and let $\Delta\vec x_{i,i-1}$ denote the corresponding phasor. 
 Then we can express $\Delta x_{i,i-1}(t)=\left|\Delta \vec x_{i,i-1}\right|\cos(\omega t+ \varphi )$ with $\varphi=\angle\left(\Delta\vec x_{i,i-1}\right)$ and 
 \[
W_i(t)=\frac{1}{2}k_i \left|\Delta \vec x_{i,i-1}\right|^2 \cos^2(\omega t+\varphi)= \frac{1}{4} k_i \left|\Delta \vec x_{i,i-1}\right|^2 \left(1+\cos(2\omega t+2\varphi)\right).
 \]
The latter expression can be written in the form
 \begin{equation}\label{eq:elastic-energies-vector-form}
     W_i(t)= \Bar{W}_{i, \mean} + \Re\left\{ \vec{W}_i e^{\jmath 2\omega t} \right\},
 \end{equation}
where the first term denotes the average elastic energy,
 \begin{equation}\label{eq:elastic-energy-phasor-mean}
     \Bar{W}_{i, \mean} = \frac{1}{4}k_i  \left|\Delta \vec{x}_
{i, i-1}\right|^2,
 \end{equation}
  and
 \begin{equation}\label{eq:elastic-energy-phasor-time}
     \vec{W}_{i}= \frac{1}{4}k_i \left(\Delta\vec{x}_{i, i-1}\right)^2
 \end{equation}
 can be interpreted as a phasor rotating with doubled frequency.  The maximal elastic energy satisfies
 \begin{align}
     W_{i, \max}  &=\Bar{W}_{i, \mean} + |\vec{W}_{i}|
     \\ &=\frac{1}{2} k_i \left| \Delta \vec{x}_{i, i-1} \right|^2 \label{eq:fatique-loading:Wimax-general} \\
      &= \frac{1}{2} k_i \left|\left(\mathbf{e}_i\tran - \mathbf{e}_{i-1}\tran\right) \mathbf{A}^{-1}(\omega) \left( \mathbf{B}_d\vec{f}_d + \mathbf{B}_a\vec{f}_a \right)\right|^{2}. \label{eq:fatique-loading:W_imax-general-long-format}
 \end{align}
 Note that the elastic potential energies in links between masses and base are solved by explicitly setting $\mathbf{e}_{0} = \mathbf{e}_{d+1} = \mathbf{o}_{d}$.
 
 When the final stage is reached, maximal values of the elastic potential energies can be simplified as given in the following proposition.
 
\begin{proposition}
%[Maximal elastic potential energies in the links]
    \label{proposition:potential-elastic-energies-links}
    When the final stage is reached, the maximal elastic potential energy in the $i^{\text{th}}$-link of the resonating subsystem is given by 
    \begin{equation}\label{eq:elastic-energies-resonating-vector-form}
         W_{i, \max} = \frac{1}{2}k_i\left|\left(\mathbf{e}\tran_{r,i} - \mathbf{e}\tran_{r, i-1}\right)\mathbf{A}_R^{-1}(\omega) \mathbf{D}_a \vec{f}_a\right|^2, \ \  1 \leq i \leq s, 
    \end{equation}
    while for the $j^{\text{th}}$-link of vibrating subsystem, it is given by
    \begin{equation}\label{eq:elastic-energies-vibrating-vector-form}
        W_{j, \max} = \frac{1}{2}k_j\left|\left(\mathbf{e}\tran_{v,j} - \mathbf{e}\tran_{v,j-1}\right)\mathbf{A}_V^{-1}(\omega) \mathbf{D}_d \vec{f}_d\right|^2, \ \ s+1 \leq j \leq d+1,
    \end{equation}
    where $\mathbf{e}_{r,i}=\begin{bmatrix}  \mathbf{o}_{i-1}\tran & 1 & \mathbf{o}_{s-i-1}\tran\end{bmatrix}\tran$ and $\mathbf{e}_{v,j}=\begin{bmatrix}  \mathbf{o}_{s-j-1}\tran & 1 & \mathbf{o}_{d-j}\tran\end{bmatrix}\tran$ encode the position of the mass $m_i$ in the resonating subsystem and the position of the mass $m_j$ in vibrating subsystem, respectively. The links at the base are handled via setting $\mathbf{e}_{r,0} = \mathbf{e}_{r,s} = \mathbf{o}_{s-1}$ for the resonating subsystem, and $\mathbf{e}_{v,s} = \mathbf{e}_{v,d+1} = \mathbf{o}_{d-s-1}$ for vibrating subsystem.
 \end{proposition}
 \begin{proof}
    %When the final stage is reached, we can split our system into resonating and vibrating subsystems; these can be treated separately. 
For resonating subsystem, i.e. $0 < i \leq s$, we rewrite \eqref{eq:prop1-proof-solved-xr} to
    \begin{equation}\label{eq:prop2-position-resonating}
        \vec{x}_i = \mathbf{e}\tran_{r,i} \mathbf{A}_R^{-1}(\omega) \mathbf{D}_a \vec{f}_a,
    \end{equation}
    which can be directly substituted to \eqref{eq:fatique-loading:Wimax-general} yielding \eqref{eq:elastic-energies-resonating-vector-form}.

    Analogously for vibrating subsystem, i.e.  $s < j \leq d+1$, we can rewrite \eqref{eq:prop1-proof-solved-xv} to
    \begin{equation}
        \vec{x}_j = \mathbf{e}\tran_{v,j} \mathbf{A}_V^{-1}(\omega) \mathbf{D}_d \vec{f}_d,
    \end{equation}
    and substitute it to \eqref{eq:fatique-loading:Wimax-general} to obtain \eqref{eq:elastic-energies-vibrating-vector-form}.
 \end{proof}

Combining results from Propositions \ref{proposition:total-absorber-force} and \ref{proposition:potential-elastic-energies-links} allows us to form the following Theorem.
 
\begin{theorem}
%[Independence of maximal elastic potential energies at the final stage on the way the force $f_a$ is generated]
    \label{remark:independence-of-fa-absorber-final-stage}
    The maximal elastic potential energies at the final stage defined in \eqref{eq:elastic-energies-resonating-vector-form} and \eqref{eq:elastic-energies-vibrating-vector-form} only depend on the disturbance force (specifically on frequency $\omega$ and amplitude $\bar f_d$) and mass, damping, and stiffness parameters within the resonating and vibrating subsystems.
\end{theorem}
\begin{corollary}
    The elastic potential energies cannot be redistributed by the control authority of the resonator. The only way for its redistribution is via a structural optimization of both the vibrating and the resonating subsystems.
\end{corollary}
% This means that no matter \emph{how} total force acting on the absorber defined in \eqref{eq:prop1-fa-normal-case} or \eqref{eq:prop1-fa-marginal-case} is generated, when the final stage is reached, we cannot influence the maximal potential energies by the generator parameters. 
% This holds true even if the structure is more complex %than linear chain of flexibly linked masses:
\begin{remark}
    Expression \eqref{eq:prop1-fa-normal-case} and \eqref{eq:prop1-fa-marginal-case} can be written in the form
    \begin{equation}\label{eq:remark-generalized-theroem}
        \vec f_a=-\mathbf{G}_1(\jmath\omega)^{-1} \mathbf{G}_2(\jmath \omega) \vec f_d
    \end{equation}
     where $\mathbf{G}_1$ and $\mathbf{G}_2$ are the transfer functions from $\vec{f}_a$ and $\vec{f}_d$ to the position of the target $\vec{x}_s$. The expression \eqref{eq:remark-generalized-theroem} is also applicable to more complex topologies than linear chains, including topologies where a separation into resonating and vibrating subsystems is not possible. Consequently, \cref{remark:independence-of-fa-absorber-final-stage} remains valid in such cases.
\end{remark}
  
\subsection{Delayed resonator analysis and design}\label{sec:dr-anaylisis-design}
% Till this point, we haven't specified fa.
Let us note that both the above propositions hold true regardless of how the force $f_a$ is generated. In what follows we consider the traditional DR with position feedback. Though, the modification to DR with velocity or acceleration feedback is straightforward. With a reference to \cref{subfig:non-collocated-scheme-overall}, the dynamics of the absorber is governed by
 \begin{equation}\label{eq:resonator-dynamics-newton-prelim}
      m_a\ddot{x}_a(t)=f_a(t),
 \end{equation}
which in the phasor form reads as
 \begin{equation}\label{eq:resonator-vector-fa}
    - m_a\omega^2 \vec{x}_a = \vec{f}_a,
 \end{equation}
 and implies
 \begin{equation}\label{eq:resonator-vector-xa}
    \vec{x}_a = -\frac{1}{m_a\omega^2}\vec{f}_a.
\end{equation}
Dynamics of DR can be obtained by substituting \eqref{eq:force} into \eqref{eq:resonator-dynamics-newton-prelim}
 \begin{equation}\label{eq:resonator-dynamics-newton}
      m_a\ddot{x}_a(t) + c_a\dot{x}_a(t)+k_ax_a(t)=c_a\dot{x}_p(t)+k_ax_p(t)+u(t),
 \end{equation}
where the delayed position feedback is considered in the form
 \begin{equation}\label{eq:resonator-ut}
    u(t) = g x_a(t-\tau),
 \end{equation}
with tuneable gain $g$ and time delay $\tau \geq 0$.
\color{black}
\begin{proposition}
%[Delayed Resonator parameters $g$, $\tau$]
    \label{proposition:DR-parameters}
    The force $f_a$ given by \eqref{eq:prop1-fa-normal-case} needed for complete vibration suppression of the mass $m_s$ excited by the force $f_d$ given by \eqref{eq:disturbance-cosine-form} and acting at the mass $m_d$ can be generated by the DR \eqref{eq:resonator-dynamics-newton}--\eqref{eq:resonator-ut} with a parameter pair $g$, $\tau$ from the sets
  \begin{align}
     & g^{+} = |Q(\omega)|, \ \ \ \  \tau^{+}=\frac{1}{\omega}(-\arg{Q(\omega)} + 2k\pi), k \in Z,
     \label{eq:DRpar1} \\
     \label{eq:DRpar2}
     & g^{-} = -|Q(\omega)|, \ \ \tau^{-}=\frac{1}{\omega}(\pi-\arg{Q(\omega)} + 2k\pi), k \in Z,
  \end{align}
with
\begin{equation}
    Q(\omega) = - m_a\omega^2 \left[1 - (\jmath\omega c_a + k_a)\left(\mathbf{e}_{r,p}\tran \mathbf{A}_R^{-1}(\omega)\mathbf{D}_a + \frac{1}{m_a\omega^2}\right) \right],
    \label{eq:Q}
 \end{equation} 
provided that the overall setup is stable and feedback with the selected pair $g$, $\tau$ is realizable, i.e. $\tau \geq 0$. %As a rule, we select the pair with the smallest positive time delay.
\end{proposition}
\begin{proof}
 %selected so that the resonating substructure (resonating subsystem with mounted DR) becomes resonant at the excitation frequency $\omega$.
 Both \eqref{eq:resonator-dynamics-newton} and \eqref{eq:resonator-ut} can be expressed in the phasor form as
 \begin{align}
     & - m_a \omega^2 \vec{x}_a = \left(\jmath\omega c_a + k_a\right) \left(\vec{x}_p - \vec{x}_a\right) + \vec{u}, \label{eq:resonator-vector-absorber} \\
     & \vec{u} = g e^{-\jmath\omega\tau} \vec{x}_a. \label{eq:resonator-vector-u-simple}
 \end{align}
From \eqref{eq:prop1-proof-solved-xr} we can express the phasor of position of mass $m_p$ as
 \begin{equation}\label{eq:resonator-vector-xp}
     \vec{x}_p = \mathbf{e}_{r,p}\tran \mathbf{A}_R^{-1}(\omega)\mathbf{D}_a \vec{f}_a,
 \end{equation}
 where $\mathbf{e}_{r,p}\tran=\begin{bmatrix}  \mathbf{o}_{p-1}\tran & 1 & \mathbf{o}_{s-p}\tran\end{bmatrix}$.
 %is vector which defines position of $\vec{x}_p$ with respect to resonating subsystem.
  Substituting \eqref{eq:resonator-vector-fa}, \eqref{eq:resonator-vector-xa} and \eqref{eq:resonator-vector-xp} into \eqref{eq:resonator-vector-absorber} and expressing $\vec{u}$ as a function of $\vec{f}_a$ gives
 \begin{equation}\label{eq:resonator-ut-function-fat}
    \vec{u} = \left[1 - (\jmath\omega c_a + k_a)\left(\mathbf{e}_{r,p}\tran  \mathbf{A}_R^{-1}(\omega)\mathbf{D}_a + \frac{1}{m_a\omega^2}\right) \right] \vec{f}_a.
 \end{equation}
 Substituting \eqref{eq:resonator-vector-xa} to \eqref{eq:resonator-vector-u-simple} yields
 \begin{equation}\label{eq:resonator-ut-function-g-tau}
    \vec{u} = - \frac{g}{m_a\omega^2} e^{-\jmath\omega\tau}\vec{f}_a.
 \end{equation}
 Comparing \eqref{eq:resonator-ut-function-fat} and \eqref{eq:resonator-ut-function-g-tau} with both sides multiplied by $-m_a\omega^2$ then gives
 \begin{equation}\label{eq:resonator-complex-g-tau-equation}
    g e^{-\jmath\omega\tau} = Q(\omega),
 \end{equation}
 where $Q(\omega)$ is of a form \eqref{eq:Q} given in proposition. Comparing the magnitude and the argument of both sides of \eqref{eq:resonator-complex-g-tau-equation} yields parameters of the DR.
 \end{proof}
\color{black}
  
In \cite{vyhlidal2019analysis}, it was derived for the collocated deployment of the DR that the best stability posture and robustness against nominal and true frequency mismatch is achieved for the first delay branch. Analogous result can be expected for the non-collocated case under consideration. Therefore, as a rule of thumb, the pair with the lowest positive delay $\tau$ is selected.
 %  Thus, we propose the following selection of the parameter set 
 % \begin{proposition}\label{proposition:smallest-positive-tau}[Selection of the $g, \tau$ parameter set]
 %  Given infinitely many possible pairs of DR parameters $(g, \tau)$ defined in Proposition \ref{proposition:DR-parameters}, as a rule of thumb, the pair with the lowest positive delay $\tau$ is selected.
 %  %This is justified due to a fact that stability regions tend to diminish as we increase delay branch $k$ \cite{vyhlidal2019analysis}.
 % \end{proposition}

%\subsubsection{Fatigue assessment in the DR's link}
Similarly to the preceding section, the potential risk of fatigue in the DR's link is quantified by the maximal elastic potential energy between the system mass $m_p$ and the absorber mass $m_a$ at the stage of equilibrium oscillations. Using \eqref{eq:elastic-energies-general-timedomain} and substituting the positions $x_a$ and $x_p$ yields
 \begin{equation}\label{eq:elastic-energies-vector-form-absorber}
     W_a(t) = \Bar{W}_{a, \mean} + \Re\left\{ \vec{W}_a e^{\jmath 2\omega t} \right\},
 \end{equation}
 where, assuming $\Delta \vec{x}_{a, p} = \vec{x}_a - \vec{x}_{p}$, time independent average is defined as
 \begin{equation}\label{eq:elastic-energy-phasor-mean-absorber}
     \Bar{W}_{a, \mean} = \frac{1}{4}k_a \left|\Delta \vec{x}_{a, p}\right|^2,
 \end{equation}
 and phasor as
 \begin{equation}\label{eq:elastic-energy-phasor-time-absorber}
     \vec{W}_{a} = \frac{1}{4}k_a \left(\Delta\vec{x}_{a, p}\right)^2,
 \end{equation}
 with maximal value
 \begin{align}
     W_{a, \max} &= \Bar{W}_{a, \mean} + |\vec{W}_{a}|, \label{eq:elastic-energy-maxima-absorber} \\
     &= \frac{1}{2}k_a\left|\left( -\frac{1}{m_a \omega^2} - \mathbf{e}\tran_{r,p}\mathbf{A}_R^{-1}(\omega) \mathbf{D}_a\right) \vec{f}_a \right|^2. \label{eq:potential-elastic-energy-absorber}
 \end{align}

\subsubsection{Actuation power assessment needed for full vibration suppression}
 The actuation power at the DR's control feedback is given by
 \begin{equation}
     p(t)=u(t)(\dot x_p(t)-\dot x_a (t)),
     \label{eq:power}
 \end{equation}
 which can not be represented by phasor, however can be rewritten to
 \begin{equation}
     p(t) = \bar{p}_{\mean} + \Re \{ \vec{p} e^{\jmath 2\omega t}\},
 \end{equation}
 where the first part of right hand side 
 \begin{equation}\label{eq:resonator-ut-phasor-time-average}
     \bar{p}_{\mean} = \frac{1}{2} \Re \{ \jmath \omega (\vec{u})^{*} \left(\vec{x}_p - \vec{x}_a\right)\}                                            
 \end{equation}
 is independent on time and determines the average value of the power, and
 \begin{equation}\label{eq:resonator-ut-phasor}
     \vec{p} = \frac{1}{2} \jmath \omega \vec{u} \left(\vec{x}_p - \vec{x}_a\right)
 \end{equation}
 oscillates with a frequency $2\omega$ with an amplitude 
 \begin{equation}
     \bar{p} = \frac{1}{2} \omega \left|\vec{u} \left(\vec{x}_p - \vec{x}_a\right) \right|.
 \end{equation}
 Thus, the maximum power needed in the steady oscillation stage is  given by 
 \begin{equation}\label{eq:resonator-actuation-power-PMAX}
     P_{\max} = \max \left\{|\bar{p}+\bar{p}_{\mean}|,|\bar{p}-\bar{p}_{\mean}|\right\}.
 \end{equation}
 % \textcolor{blue}{
 % which we convert to phasor form and substitute \eqref{eq:resonator-ut-function-fat}, \eqref{eq:resonator-vector-xp} and \eqref{eq:resonator-vector-xa} to obtain
 % \begin{equation}\label{eq:resonator-actuation-power-vector}
 %     \vec{p}(t) = \mathrm{j}\omega \left[\left(\mathbf{e}_{r,p}\tran  \mathbf{A}_R^{-1}(\omega)\mathbf{D}_a + \frac{1}{m_a\omega^2}\right) - (\mathrm{j}\omega c_a + k_a)\left(\mathbf{e}_{r,p}\tran  \mathbf{A}_R^{-1}(\omega)\mathbf{D}_a + \frac{1}{m_a\omega^2}\right)^2 \right] \left(\vec{f}_a(t)\right)^2.
 % \end{equation}
 % }
 
\subsection{Stability analysis}
 %Mounting VA, $m_a$, $k_a$, $c_a$, complemented with control force \eqref{eq:resonator-ut} on mass $m_p$ ensures resulting force $f_a(t)$ from proposition \ref{proposition:total-absorber-force} acts on our setup in a way mass $m_s$ is stopped, however it does not ensure stability of the overall system. Until this point, it was beneficial for us to think in terms of $d$ second-order differential equations. However, for examining stability of the designed system, it is useful to convert the system representation to the state space representation by extending the vector of displacements from \eqref{eq:system-general-MCK} to also contain velocities as
 The necessary condition for achieving the entire vibration suppression at the mass $m_s$ is the stability of the overall system consisting of the series of linked masses and the DR. The stability of the whole system can be checked by determining the rightmost characteristic roots of the system, e.g. by MATLAB package \emph{TDS-CONTROL} \cite{appeltans2022tds}, \cite{appeltans2023analysis}. For that purpose, the system model is to be turned to the set of Delay Differential Algebraic
Equations (DDAE). First, redefine the state vector of \eqref{eq:system-general-MCK} as  
 \begin{equation}\label{eq:stability-overline-x-vector}
    \overline{\mathbf{x}}(t) = \begin{bmatrix} \mathbf{x}(t)\tran & \dot{\mathbf{x}}(t)\tran \end{bmatrix}\tran,
 \end{equation}
 and rewrite \eqref{eq:system-general-MCK} to
 \begin{equation}\label{eq:stability:extended-x-system}
     \begin{bmatrix}
         \mathbf{I} & \mathbf{0} \\
         \mathbf{0} & \mathbf{M}         
     \end{bmatrix} \dot{\overline{\mathbf{x}}}(t)
     =
     \begin{bmatrix}
         \mathbf{0} & \mathbf{I} \\
         -\mathbf{K} & -\mathbf{C}         
     \end{bmatrix} \overline{\mathbf{x}}(t)
     +
     \begin{bmatrix}
         \mathbf{o}_{d} \\ \mathbf{B}_a
     \end{bmatrix} f_a(t)
     +
     \begin{bmatrix}
         \mathbf{o}_{d} \\ \mathbf{B}_d
     \end{bmatrix} f_d(t).
 \end{equation}
 %Secondly, substituting  \eqref{eq:resonator-dynamics-newton}
 Analogously, define the state vector of the absorber as \begin{equation}\label{eq:stability-overline-xa-vector}
   \overline{\mathbf{x}}_a(t) = \begin{bmatrix} x_a(t) & \dot{x}_a(t) \end{bmatrix}\tran,
 \end{equation}
 %we recall \eqref{eq:resonator-dynamics-newton}, from which we can obtain the total force acting on the absorber.
 %\begin{equation}\label{eq:stability-fa-time-domain}
 %    c_a(\dot{x}_p(t)-\dot{x}_a(t))+k_a(x_p(t)-x_a(t))+u(t) = f_a(t),
 %\end{equation}
 and rewrite \eqref{eq:resonator-dynamics-newton-prelim} to 
 \begin{equation}\label{eq:stability:extended-xa-absorber}
     \begin{bmatrix}
         1 & 0 \\
         0 & m_a
     \end{bmatrix} \dot{\overline{\mathbf{x}}}_a(t)
    =
     \begin{bmatrix}
         0 & 1 \\
         0 & 0
     \end{bmatrix} \overline{\mathbf{x}}_a(t)
     +
     \begin{bmatrix} 0 \\ 1 \end{bmatrix} f_a(t).
 \end{equation}
 Further on, equation  \eqref{eq:force} can be turned to the form 
 \begin{equation}\label{eq:stability:total-force-fa}
     0 = \begin{bmatrix} k_a \mathbf{e}_p\tran & c_a \mathbf{e}_p\tran \end{bmatrix} \overline{\mathbf{x}}(t) + \begin{bmatrix}-k_a & -c_a  \end{bmatrix}\overline{\mathbf{x}}_a(t) + u(t) - f_a(t),
 \end{equation}
 and the control law \eqref{eq:resonator-ut} can be turned to
 \begin{equation}\label{eq:stability:control-law}
     0 = g  \begin{bmatrix} 1 & 0 \end{bmatrix} \overline{\mathbf{x}}_a(t-\tau) - u(t).
 \end{equation}

Defining the extended state vector
 \begin{equation}
    \Tilde{\mathbf{x}}(t) = \begin{bmatrix} \mathbf{x}\tran(t) & \dot{\mathbf{x}}\tran(t) & x_a(t) &  \dot{x}_a(t) & f_a(t) & u(t)  \end{bmatrix}\tran,
 \end{equation}
 the DDAE model for \eqref{eq:stability:extended-x-system}, \eqref{eq:stability:extended-xa-absorber}, \eqref{eq:stability:total-force-fa} and \eqref{eq:stability:control-law} is as follows,
 \begin{equation}\label{eq:stability:DDAE}
    \mathcal{E}\dot{\Tilde{\mathbf{x}}}(t) = \mathcal{A}_0 \Tilde{\mathbf{x}}(t) + \mathcal{A}_1 \Tilde{\mathbf{x}}(t-\tau) + \mathcal{B}f_d(t),
 \end{equation}
 where
 \begin{equation*}
     \mathcal{E} = \text{diag} (\mathbf{I}, \mathbf{M}, 1, m_a, 0, 0),\ \mathcal{B}
    =
    \begin{bmatrix}
        \mathbf{o}_{d}\tran & \mathbf{B}_d\tran & 0 & 0 & 0 & 0
    \end{bmatrix}\tran,
 \end{equation*}
 %and matrices $\mathcal{A}_0$, $\mathcal{A}_1$ are defined as 
 \begin{equation*}
    \mathcal{A}_0
     =
     \begin{bmatrix}
         \mathbf{0} & \mathbf{I} & \mathbf{o}_{d} & \mathbf{o}_{d} & \mathbf{o}_{d} & \mathbf{o}_{d}\\
         -\mathbf{K} & -\mathbf{C} & \mathbf{o}_{d} & \mathbf{o}_{d} & \mathbf{B}_a & \mathbf{o}_{d}\\
         \mathbf{o}_{d}\tran & \mathbf{o}_{d}\tran & 0 & 1 & 0 & 0  \\
         \mathbf{o}_{d}\tran & \mathbf{o}_{d}\tran & 0 & 0 & -1 & 0\\
         k_a \mathbf{e}_p\tran & c_a \mathbf{e}_p\tran & -k_a & -c_a & -1 & 1\\
         \mathbf{o}_{d}\tran & \mathbf{o}_{d}\tran & 0 & 0 & 0 & -1\\
     \end{bmatrix}, \ 
     \mathcal{A}_1
     =
     \begin{bmatrix}
         \mathbf{0} & \mathbf{0} & \mathbf{o}_{d} & \mathbf{o}_{d} & \mathbf{o}_{d} & \mathbf{o}_{d}\\
         \mathbf{0} & \mathbf{0} & \mathbf{o}_{d} & \mathbf{o}_{d} & \mathbf{o}_{d} & \mathbf{o}_{d}\\
         \mathbf{o}_{d}\tran & \mathbf{o}_{d}\tran & 0 & 0 & 0 & 0  \\
         \mathbf{o}_{d}\tran & \mathbf{o}_{d}\tran & 0 & 0 & 0 & 0\\
         \mathbf{o}_{d}\tran & \mathbf{o}_{d}\tran & 0 & 0 & 0 & 0\\
         \mathbf{o}_{d}\tran & \mathbf{o}_{d}\tran & g & 0 & 0 & 0\\
     \end{bmatrix}.
 \end{equation*}
 
 The dynamic behavior of the controlled system \eqref{eq:stability:DDAE} is governed by
 \begin{equation}\label{eq:stability:DDAE-determinant}
     P(\lambda) = \det \left( \mathcal{E} \lambda - \mathcal{A}_0 - \mathcal{A}_1 e^{-\lambda \tau} \right),
 \end{equation}
which in general has infinitely many roots. The stability is implied by the location of all the roots in the left hand side of complex plane. Such a condition can be defined in a compact way via introducing the \emph{spectral abscissa}
\begin{equation}\label{eq:stability:spectral-abscissa}
    \alpha = \max_{\lambda \in \mathbb{C}} \{ \Re \left(\lambda \right): \ P(\lambda)=0\},
\end{equation}
and requiring $\alpha < 0$.

\section{Simultaneous structural optimization and control design}
\label{section:composed-structural-optimization}

As it results from the above analysis, with a reference to \cref{fig:non-collocated-scheme}, once the system is being excited by a harmonic force \eqref{eq:disturbance-cosine-form} acting on the mass $m_d$ and the mass $m_s$ is to be stopped, there exists a unique control force $f_a$ acting at the mass $m_p$ given in \cref{proposition:total-absorber-force}. Thus, as stated in \cref{remark:independence-of-fa-absorber-final-stage}, the motion of masses in both the resonating and vibrating subsystems cannot be influenced by the control authority and they are fully determined by the structural parameters of the setup itself. Consequently, the only way of modifying the potential energy in the links is via the adjustment of parameters $m_i,\ i=1, \dots, d,\ i\ne s$, and $c_i,\ k_i,\ i=1, \dots, d+1$. 

Besides, considering the two-parameter control rule \eqref{eq:resonator-ut}, the parameters of which are by \cref{proposition:DR-parameters}, no degrees of freedom in the DR feedback is left for tuning the closed-loop dynamics. Unlike in \cite{PilbauerRobust},  %\cite{silmStabilizationZeroLocation}, 
\cite{silm_2024_spectral_design}, \cite{Saldanha2023}, \cite{kuvre2018spectral}, where either the set of free parameters in the DR feedback was enlarged or a supplementary stabilizing controller was deployed, we stick to the simple structure DR position feedback \eqref{eq:resonator-ut}, which is easy to implement and tune in practice. Thus, the only way to ensure the stability of the considered setting is to include the stability condition in the structural parameter optimization.
%\color{red}
Also note that the stability constraint would be analogous if we considered DR velocity feedback instead of the DR position feedback. In both cases, the DDAE model would be of retarded type. Considering the acceleration feedback, however, would lead to neutral DDAE. Additionally to standard stability check of the DDE, the strong stability check of the associated difference equation would be needed. %\cite{hale2003stability}. 
%\color{black}
% Besides, considering the two-parameter control rule \eqref{eq:resonator-ut}, the parameters of which are tuned by Proposition \ref{proposition:DR-parameters} to achieve the vibration suppression at the mass $m_s$, no degree of DR parameters is left for tuning the closed loop dynamics, except the possibility to choose either the setting with either the positive \eqref{eq:DRpar1} or negative \eqref{eq:DRpar2} gain, and selection of delay branch $k$. In \cite{vyhlidal2019analysis}, it has been derived for the collocated deployment of the DR that the best stability posture is achieved for the first delay branch. Analogous result can be expected for the noncollocated case considered here. Thus, we assume that parameter set with a smallest possible delay $\tau$ is used in \eqref{eq:resonator-ut}. 
% Unlike in \cite{PilbauerRobust}, \cite{kuvre2018spectral}, \cite{silmStabilizationZeroLocation}, \cite{Saldanha2023}, where either parameter set in the DR feedback was extended or supplementary stabilizing controller was applied, we stick on the simple structure DR position feedback \eqref{eq:resonator-ut}, which is easy to implement, and include the stability condition to the structural parameter optimization.   

An additional quantity to follow is the power needed for beating the vibration at the non-collocated mass $m_s$. As it results from \eqref{eq:power}, it is determined by a product of velocity difference of masses $m_p$ and $m_a$ and the actuation force of the DR given by \eqref{eq:resonator-ut}. Hence, it can be influenced by the parameters of the whole resonating substructure, consisting of the resonating subsystem with masses $m_i,\ i=1, \dots, s-1$ and the DR with mass $m_a$.  
Additional constraints may arise from the limits on the construction parameters, and safety and durability of the DR which may be more relaxed compared to those characteristics of the system.

To sum up the discussion above, the following objectives and constraints are to be followed in the structural parameter optimization:
\begin{enumerate}
    \item {\em Full vibration suppression at the non-collocated mass $m_s$}. This task is directly fulfilled by the equality constraints defined in Proposition~\ref{proposition:DR-parameters} and taking the pair of parameters $g, \tau$ with the smallest positive delay $\tau$.
    \item {\em Minimizing the risk of fatigue across the system}. The fatigue is quantified via a potential energy in the system links derived in Proposition~\ref{proposition:potential-elastic-energies-links}. The objective is to minimize this quantity at the most risky link, which leads to the solution of the min-max problem:
    \begin{equation}\label{eq:optimization:criterion-maxW}
        \min\{\max\{ W_{i, \max}: \ 1\leq i\leq d \}\}.
    \end{equation}
    \item {\em Minimizing the maximum of the power needed for the vibration suppression}. This task directly leads to
    \begin{equation}\label{eq:optimization:criterion-maxP}
        \min\{P_{\max}\}, 
    \end{equation}
    where $P_{\max}$ is given by \eqref{eq:resonator-actuation-power-PMAX}.
    \item {\em Stability of the overall system}. This can be achieved by considering the constraint
    \begin{equation}\label{eq:constraint-abscissa-stable}
        \alpha \leq \xi_{\alpha}
    \end{equation}
    where the spectral abscissa $\alpha$ of the DDAE system \eqref{eq:stability:DDAE} is given by \eqref{eq:stability:spectral-abscissa}, and $\xi_{\alpha}<0$, is a selected distance from the stability boundary, imposing a desired stability margin.
    \item {\em Physical parameter constraints.} Limits on the selected parameters to be optimized (mass, damping and stiffness of the links).
    \item {\em A DR safety and durability constraint.} Assessing the fatigue in the absorber's link in terms of potential energy \eqref{eq:potential-elastic-energy-absorber}, it leads to the condition
    \begin{equation}\label{eq:constraint-max-Wa}
      W_{a, \max} \leq \xi_a  
    \end{equation}
    where $\xi_a$ is the predefined limit. 
\end{enumerate}
The above considerations lead us to the optimization problem formulation 
  \begin{mini!}|l|[2]
    {g, \tau, \mathbf{\theta}}{ \gamma \frac{\max\{ W_{i, \max}(\theta):\ i=1,\dots,d \}}{W_{\nom}}  + (1-\gamma) \frac{P_{\max}(g, \tau, \theta)}{P_{\nom}}} 
    {}{\label{mini:multiobject-minimisation-problem}}
    \addConstraint{\mathbf{A}\mathbf{\theta} \leq \mathbf{b} \label{eq:constraint:Ax<=b}}
    \addConstraint{\alpha (g, \tau, \theta)\leq \xi_{\alpha} \label{eq:constraint:abscissa-stable}}
    \addConstraint{W_{a, \max}(g, \tau, \mathbf{\theta}) \leq \xi_a \label{eq:constraint:Wa-maximum}}
    \addConstraint{g, \tau \text{ fullfill Proposition \ref{proposition:DR-parameters}}, \label{eq:constraint:gtau-fullfills-prop}}
  \end{mini!}
 where $g, \tau$ are control law parameters and $\mathbf{\theta}$ is the structural parameter set. The parameter $\gamma \in \langle0,1\rangle$ allows us to balance the two objectives \eqref{eq:optimization:criterion-maxW} and \eqref{eq:optimization:criterion-maxP} according to user-defined priority. Note that the two objectives are to be normalized with respect to nominal values $P_{\nom}$ and $W_{\nom}$. The linear inequality constraints \eqref{eq:constraint:Ax<=b} represent physical feasibility limits such as positivity of all parameters and their lower and upper bounds, and possible other specific structural requirements. An example of the latter is a restriction placed on absorber mass $m_a$ to be lower than mass $m_p$. The stability margin of the overall setup is represented via \eqref{eq:constraint:abscissa-stable}, which is a nonlinear inequality constraint formed in \eqref{eq:constraint-abscissa-stable}. %\textcolor{blue}{For computation of value and gradient of spectral abscissa, we rely solely on MATLAB package \emph{TDS-Control} \cite{appeltans2022tds, appeltans2023analysis} and since change in parameters $\theta$ affects time delay $\tau$, numerical gradient is used}. 
 The nonlinear inequality constraint \eqref{eq:constraint:Wa-maximum} poses a limit on maximum of elastic potential energy $W_{a, \max}$ in DR's link, representing the maximum possible load we can redirect from the structure to the absorber. Finally, constraint \eqref{eq:constraint:gtau-fullfills-prop} ensures full vibration suppression of mass $m_s$.

 Note that variables $g$ and $\tau$ and constraint \eqref{eq:constraint:gtau-fullfills-prop} can be eliminated, using \cref{proposition:DR-parameters}. This is due to the fact that $g$ and $\tau$ are determined from structural parameters including those in $\mathbf{\theta}$, and thus they are not free. In fact, by choosing the solution (\ref{eq:DRpar1})--(\ref{eq:DRpar2}) with smallest positive delay, they are uniquely determmined by $\mathbf{\theta}$.  Therefore, solving the above optimization task can be rephrased as
 \begin{mini!}|l|[2]
    {\mathbf{\theta}}{ \gamma \frac{\max\{ W_{i, \max}(\theta):\ i=1, \dots, d \}}{W_{\nom}}  + (1-\gamma) \frac{P_{\max}(\theta)}{P_{\nom}}} 
    {}{\label{mini:multiobject-minimisation-problem-transformed}}
    \addConstraint{\mathbf{A}\mathbf{\theta} \leq \mathbf{b} \label{eq:constraint2:Ax<=b}}
    \addConstraint{\alpha (\theta)\leq \xi_{\alpha} \label{eq:constraint2:abscissa-stable}}
    \addConstraint{W_{a, \max}(\mathbf{\theta}) \leq \xi_a \label{eq:constraint2:Wa-maximum}.}
  \end{mini!}
  Since the objective function \eqref{mini:multiobject-minimisation-problem-transformed}, as well as some of the constraint functions, are in general nonconvex and nonsmooth, makes this optimization problem a suitable candidate for an application of the software \emph{GRANSO}  \cite{curtis2017bfgs}, which implements a BFGS-SQP method for solving nonsmooth, nonconvex, constrained optimization problems. The application of this method will be demonstrated in the subsequent numerical case study. 
 \color{black}
 Before that, we present experimental case study, where we demonstrate the possibility to improve system characteristics towards non-collocated vibration absorption considering two-parameter set for which they can well be mapped and no optimization routine is needed.  
\color{black}
\section{Structural analysis and optimization of experimental setup}
\label{section:case-study-experimental}
 In order to explain
the mechanisms by which the optimized parameters influence the optimization objectives, we analyse and optimize parameters of the experimental set-up depicted in Fig. \ref{fig:experimental-setup}. It consists of flexibly linked three masses with configuration $p=1$, $s=2$ and $d=3$, meaning that the absorber is deployed at the mass $m_1$ and the objective is to suppress vibration of the middle mass $m_2$. The carts $m_1$, $m_2$ and $m_3$ are attached by industrial ball bearings to the rail fixed on the frame. The absorber $m_a$ is mounted directly on a cart $m_1$ where it slides on a smaller rail. All carts are flexibly interconnected by springs. A multi-pole magnetic strip is installed on the setup frame to measure the displacement of the carts. The Moticont LVCM-032-076-20 voicecoil is used as the DR to exert the force input $u(t)$, while Akribys AVM40-20-0.5 voicecoil is used to generate the disturbance harmonic force $f(t)$. 
The control algorithms and instrumentation are implemented in LabVIEW\textsuperscript{\texttrademark} 2021 and are executed on the NI compactRIO 9064 industrial control system from National Instruments with a sampling rate of \SI{1}{kHz}. More detailed description of the setup hardware can be found in \cite{silm_2024_spectral_design}, \cite{silm_2024_spectral_design}, and \cite{kure2024robust}, where, in different configurations, it was used for validating different vibration suppression problems. 

The harmonic force \eqref{eq:disturbance-cosine-form} acting at mass $m_3$ is characterized by $\bar f_d=\SI{2}{N}$ and frequency $\omega=\SI{26.389}{s^{-1}} (\SI{4.2}{Hz})$.
 The nominal parameters of the setup structure are assumed as: $m_1=\SI{1.49}{kg}$, $m_2=\SI{0.509}{kg}$, $m_3=\SI{0.687}{kg}$, $m_a=\SI{0.42}{kg}$, $k_1=\SI{1001}{Nm^{-1}}$, $k_2=\SI{749}{Nm^{-1}}$, $k_3=\SI{711}{Nm^{-1}}$, $k_4=\SI{950}{Nm^{-1}}$, $k_a=\SI{407}{Nm^{-1}}$, $c_1=\SI{4.35}{Nsm^{-1}}$, $c_2=\SI{0.85}{Nsm^{-1}}$, $c_3=\SI{1.85}{Nsm^{-1}}$, $c_4=\SI{4.95}{Nsm^{-1}}$, $c_a=\SI{1.8}{Nsm^{-1}}$.
 \begin{figure}[t]
    \centering
    \includegraphics[width=\textwidth]{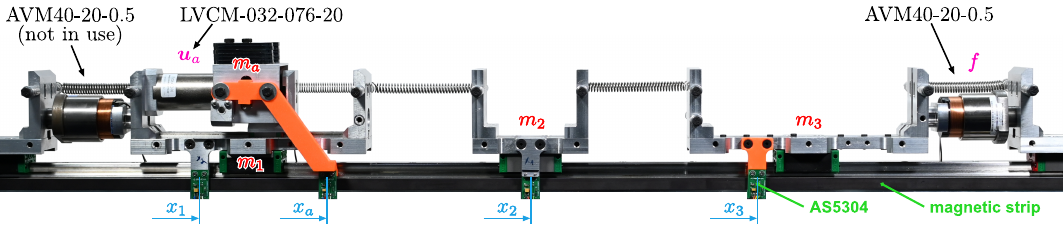}
    \caption{\textcolor{black}{Experimental setup with three flexibly linked carts actuated by the voice coils}}
    \label{fig:experimental-setup}
 \end{figure}
 Considering the parameters of the decomposed model \eqref{eq:prop1-proof-resonating}-\eqref{eq:prop1-proof-vibrating}, we have $\mathbf{A}_R=-\omega^2m_1 + k_1+k_2 +\jmath\omega(c_1+c_2)$, $\mathbf{a}_R=-k_2 -\jmath\omega c_2$, $a_{s,s} = -\omega^2m_2 + k_2+k_3 +\jmath\omega(c_2+c_3)$, $\mathbf{A}_V=-\omega^2m_3 + k_3+k_4 +\jmath\omega(c_3+c_4)$, $\mathbf{a}_V=-k_3 -\jmath\omega c_3$, $\mathbf{D}_a=1$ and $\mathbf{D}_d=1$.

 Regarding the vector of decisive parameters, we opt for adjusting the masses as the change is easily measurable and can be done by either adding steel plates on the carts or removing them. Before selection of masses to be optimize, the structural analysis is performed.
 %We select $\theta = \begin{bmatrix} m_a & m_3 \end{bmatrix}$, with motivation for picking these emerging for the following analysis.
 When the vibration of $m_2$ is fully suppressed, the total force acting on absorber \eqref{eq:prop1-fa-normal-case} yields
 \begin{equation}
     \vec{f}_a = -\frac{(-k_3 - c_3\jmath\omega)(-m_1\omega^2 +k_1 + k_2 +(c_1+c_2)\jmath\omega)}{(-k_2 - c_2\jmath\omega)(-m_3\omega^2 + k_3 + k_4 +(c_3+c_4)\jmath\omega)} \vec{f}_d.
 \end{equation}
 Expressing  
 \begin{align}
     \vec{x}_1 &= \frac{(-k_3 - c_3\jmath\omega)}{(-k_2 - c_2\jmath\omega)(-m_3\omega^2 + k_3 + k_4 +(c_3+c_4)\jmath\omega)} \vec{f}_d,\\
     \vec{x}_3 &= -\frac{1}{-m_3\omega^2 + k_3 + k_4 +(c_3+c_4)\jmath\omega} \vec{f}_d.
 \end{align}
from \eqref{eq:prop1-proof-solved-xr} and \eqref{eq:prop1-proof-solved-xv}, respectively, and substituting them to  \eqref{eq:fatique-loading:Wimax-general} the maxima of the structural potential elastic energies in the links are given as
 % can be expressed via \eqref{eq:fatique-loading:Wimax-general}. As they are dependent on the corresponding link stiffness and the square of a difference of corresponding position phasors, it is clear that the only mass affecting this objective is $m_3$ (See Fig. \ref{fig:experiment-criterion} - top left). For instance, expressing maxima of potential elastic energy in the second link yields
 % \begin{equation}
 %     W_{2,max} = \frac{1}{2} k_2 \frac{(k_3^2 + c_3^2\omega^2)}{(k_2^2 + c_2^2\omega^2)((-m_3\omega^2 + k_3 + k_4)^2 + (c_3+c_4)^2\omega^2)} |\vec{f}_d|^2.
 % \end{equation}
  \begin{align}
      W_{1,max} &= \frac{1}{2} k_1 \frac{(k_3^2 + c_3^2\omega^2)}{(k_2^2 + c_2^2\omega^2)((-m_3\omega^2 + k_3 + k_4)^2 + (c_3+c_4)^2\omega^2)} |\vec{f}_d|^2, \\
      W_{2,max} &= \frac{1}{2} k_2 \frac{(k_3^2 + c_3^2\omega^2)}{(k_2^2 + c_2^2\omega^2)((-m_3\omega^2 + k_3 + k_4)^2 + (c_3+c_4)^2\omega^2)} |\vec{f}_d|^2, \\
      W_{3,max} &= \frac{1}{2} k_3 \frac{1}{(-m_3\omega^2 + k_3 + k_4)^2 + (c_3+c_4)^2\omega^2} |\vec{f}_d|^2,\\
      W_{4,max} &= \frac{1}{2} k_4 \frac{1}{(-m_3\omega^2 + k_3 + k_4)^2 + (c_3+c_4)^2\omega^2} |\vec{f}_d|^2.
  \end{align}
 As can be seen, the only mass which affects their value is $m_3$, which will be considered as the first decisive variable for the analysis.  
 %In fact, in this particular set-up and considering only masses to be changed, optimal value for potential elastic energy maxima can be obtained by maximizing the term $((-m_3\omega^2 + k_3 + k_4)^2 + (c_3+c_4)^2\omega^2)$ as it is the only term with $m_3$ and is present in all denumenators.
The additional quantity to be observed is the maximum of potential energy at the absorber $W_{a,max}$, given by \eqref{eq:potential-elastic-energy-absorber}. Expressing
 \begin{equation}
     \vec{x}_a = \frac{(-k_3 - c_3\jmath\omega)(-m_1\omega^2 +k_1 + k_2 +(c_1+c_2)\jmath\omega)}{(m_a\omega^2)(-k_2 - c_2\jmath\omega)(-m_3\omega^2 + k_3 + k_4 +(c_3+c_4)\jmath\omega)} \vec{f}_d,
 \end{equation}
form \eqref{eq:resonator-vector-xa}, it is clear that due to dependence of $W_{a,max}$ on $\Delta\vec{x}_{a,p}=\vec{x}_a-\vec{x}_p$, it depends not only at $m_3$, but also on $m_a$ and $m_1$. This is also the case for $P_{max}$ given by \eqref{eq:resonator-actuation-power-PMAX} depending on the product of force 
 \begin{equation}
\vec{u} = -m_a\omega^2 - (\jmath\omega c_a +k_a)(\vec{x}_p - \vec{x}_a)
\end{equation}
expressed from \eqref{eq:resonator-vector-absorber} and phasor difference of velocities $\jmath\omega(\vec{x}_p - \vec{x}_a)$. Note that full expressions for $W_{a,max}$ and $P_{max}$ are omitted due to their extensive lengths. Due to visualization and demonstration purposes, we select $m_a$ as the second decisive variable, while $m_1$ is considered as fixed, implying $\theta=[m_a, m_3]\tran$. Note also that all the masses play role in evaluating the spectral abscissa \eqref{eq:stability:spectral-abscissa} under the active delayed feedback \eqref{eq:resonator-ut} with parameters given by \eqref{eq:DRpar1}. 
 \begin{table}[t]
 \color{black}
      \centering
      \begin{tabular}{cc|cccc}
        variable & units & lower bound & upper bound & nominal & optimized \\
        $\theta$ & & $\theta^{\mathrm{lb}}$ & $\theta^{\mathrm{ub}}$ & $\theta^{0}$ & $\theta^{*}$ \\
        \hline
           $m_a$ & $\SI{}{kg}$ &0.220 & 0.620 & 0.420 & 0.520 \\
           $m_3$ & $\SI{}{kg}$ &0.705 & 1.205 & 1.110 & 0.705 \\
      \end{tabular}
      \caption{\textcolor{black}{Nominal and optimized values of decision variables $\theta$, together with lower and upper bounds.}}
      \label{tab:experiment-decision-variables}
 \end{table}
 \begin{table}[t]
 \color{black}
\centering
     \begin{tabular}{l|ll|ll|ll}
          & \multicolumn{2}{c|}{Objectives} & \multicolumn{2}{c|}{Inequality constraints} & \multicolumn{2}{c}{Control law} \\
          \multicolumn{1}{c|}{setup} &  $\max\limits_{i}W_{i,\max}~\text{[J]}$ & $P_{\max}~\text{[W]}$ & $\alpha~\text{[}\SI{}{s^{-1}}\text{]}$ &  $W_{a, \max}~\text{[J]}$ & $g~\text{[}\SI{}{Nm^{-1}}\text{]}$ &  $\tau~\text{[s]}$ \\ \hline
          nominal & 0.00232 & 0.03129 & -0.62415 & 0.00205 & -78.05282 & 0.03303 \\
          %nominal-measured & 0.00294 &  0.04177 & stable & 0.00276 & --- & --- \\ \hline
          optimized & 0.00136 & 0.01855 & -0.56855 & 0.00058 & -170.99583 & 0.01421 \\
          %optimized-measured & 0.00164 & 0.02321 & stable & 0.00006 & --- & --- \\ \hline
     \end{tabular}
     \caption{\textcolor{black}{Comparison of objective function values, and inequality constraints. Results presented for both nominal and optimized parameters.}}
     \label{tab:experiment-criterion}
\end{table}
 \begin{figure}[ht!]
    \centering
    \includegraphics{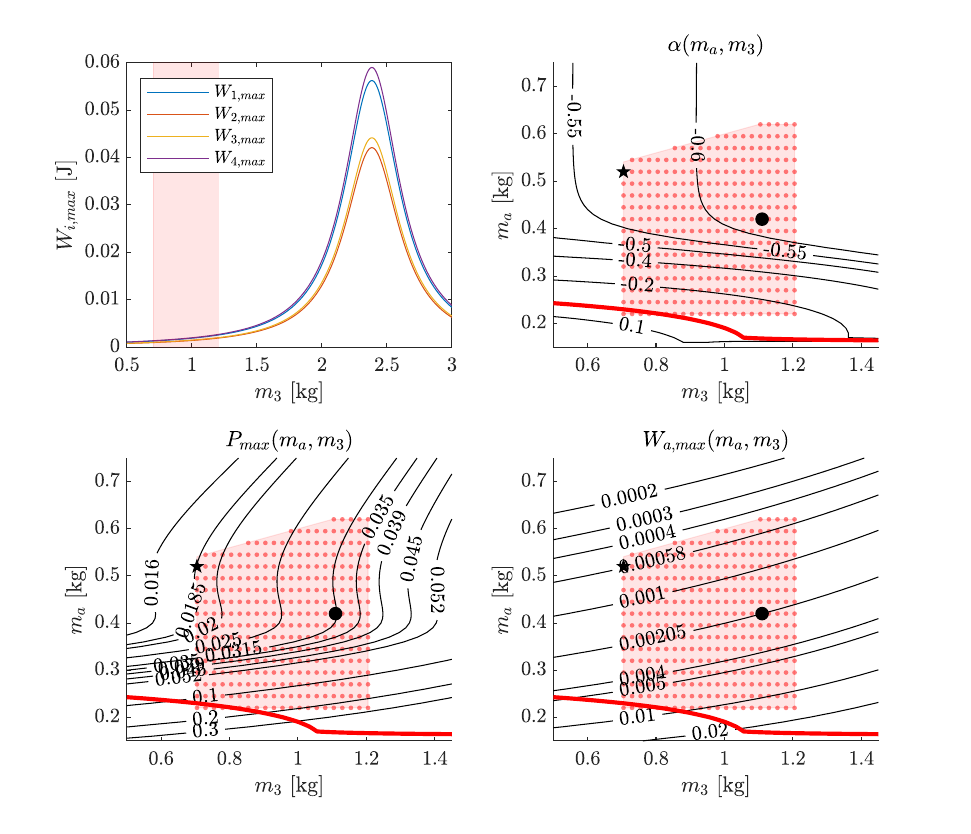}
    \caption{\textcolor{black}{Left column --- objectives, maxima of potential elastic energies in links $W_{i, max}$ as a function of $m_3$ (top), maxima of power as a function of $m_3$ and $m_a$. Right column --- non linear inequalities as a function of $m_3$ and $m_a$, spectral abscissa (top) and maxima of potential elastic energy in absorber link (bottom). Red polygons define structural constraints (linear inequalities), small red dots represent considered grid with $\SI{0.025}{kg}$ steps. Thick red line represents stability boundary, i.e. $\alpha(m_3, m_a) = 0$. Plane graphs contain marker for nominal (circle) and optimized (star) solutions.}}
    \label{fig:experiment-criterion}
 \end{figure}
 \begin{figure}[ht!]
    \centering
    \includegraphics{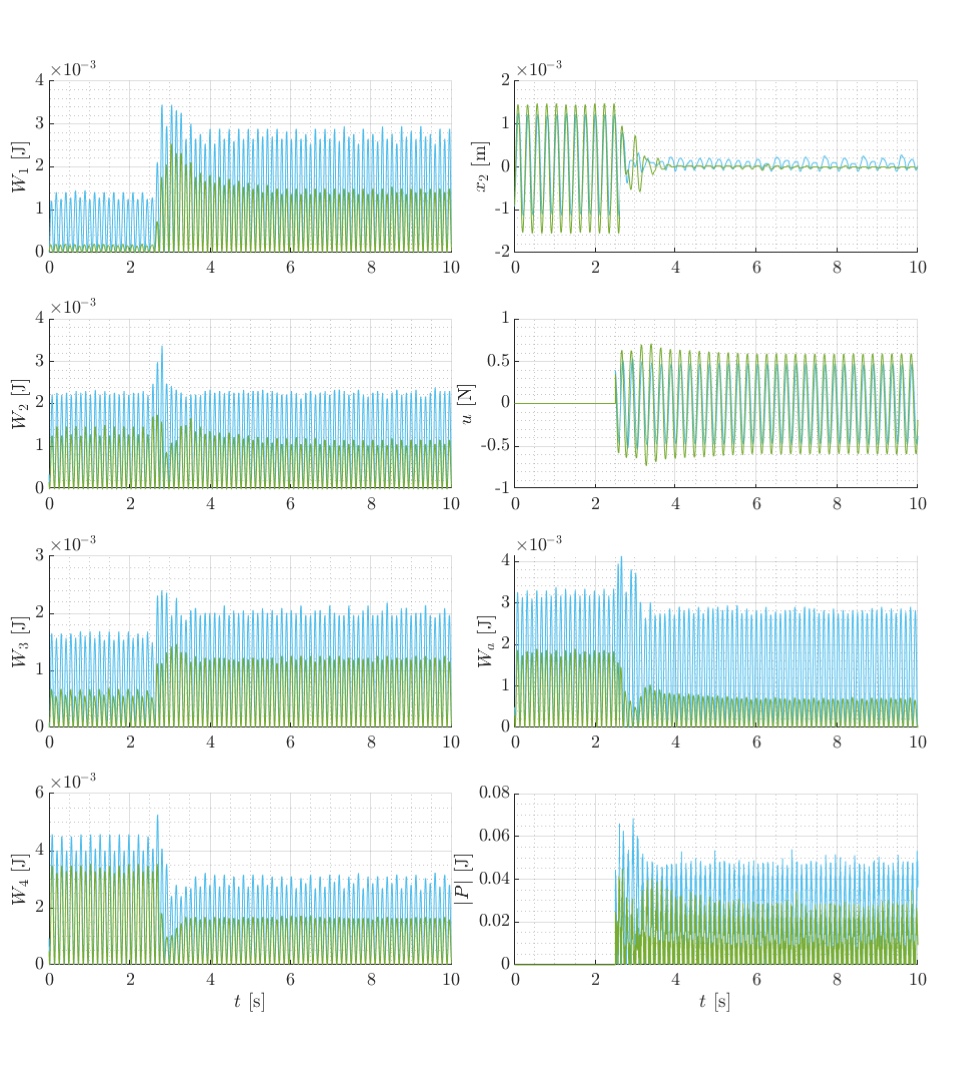}
    \caption{\textcolor{black}{Left --- Measured transients of potential elastic energies in links $i=1, \dots, 4$ Right --- Measured transients (from top to bottom): position of target mass $m_2$, control force $u$, potential elastic energy in absorber link $W_a$, absolute value of power $|P$. For all plots: Nominal parameters (blue) and optimized parameters (green). Passive regime for $t\in[0, 2.5)\, \mathrm{s}$; active regime $t\in[2.5, 10]\, \mathrm{s}$.}}
    \label{fig:experiment-elastic-energies-absorber}
\end{figure}

The analysis of maxima of potential energies $W_{i,max}, i=1..4$ and $W_{a,max}$, spectral abscissa $\alpha$, and power maximum $P_{max}$ are shown in Fig. \ref{fig:experiment-criterion}. The results were obtained numerically by evaluating the observed quantities at a dense grid and using the plotting functionalities of Matlab. As can be seen in the top-left sub-figure, the dependence of all $W_{i,max}, i=1..4$ on $m_3$ is substantial, with maxima achieved for $m_3=\SI{2.385}{kg}$. Relatively large decay can be observed for both higher and smaller values of $m_3$ compared to the maximum. In the top-right sub-figure, the spectral abscissa $\alpha$ is visualized by contour plot. The stability boundary (where $\alpha=0$) is highlighted by thick red line and transferred to the bottom sub-figures, where $P_{max}$ and $W_{a,max}$ are contoured. In all the sub-figures, the admissible parameter ranges according to Table \ref{tab:experiment-decision-variables} are visualized by red regions. Additional to lower and upper bounds of the masses given by construction constraints, we require $m_a$ to be less then 20\% of $m_1 + m_2 + m_3$.  
%This structural feasibility region is visualized via red region on Fig. \ref{fig:experiment-criterion}. 
The grid points within the red regions correspond to possible mass adjustments, as they only be changed by either adding or removing plates of $\SI{0.025}{kg}$ weights. Besides, the nominal and optimized values are highlighted by black circle and star, respectively. As can be seen, both the partial objectives $\max\limits_{i}W_{i,\max}~\text{[J]}$ (maximum achieved for $i=4$) and $P_{\max}~\text{[W]}$ are minimal for the same mass value set given in Table \ref{tab:experiment-decision-variables}, see also Table \ref{tab:experiment-criterion} for comparison of objectives for nominal and optimized parameters and the control feedback gain-delay set. As given in Table \ref{tab:experiment-criterion}, the value of $W_{a,max}$ was substantially reduced, while $\alpha$ changed only slightly. 

To complete the analysis, the analytical results are confirmed by experimental results shown in Fig. \ref{fig:experiment-elastic-energies-absorber}. Note that the passive regime can be seen in the interval $t\in[0, 2.5]\rm{s}$ after which the delayed position feedback of DR is turned on resulting to almost complete vibration suppression at the target mass $m_2$ after the short-time transient. As can be seen, substantially better performance was achieved with the optimized parameters compared to the nominal parameter set. Thus, the experiment confirms the analytically obtained results.

\color{black}
\section{Numerical complex model case study}
\label{section:case-study}
 In order to validate the \textcolor{black}{above proposed optimization based design, within this numerical case study}, we consider the system configuration with $d=5$, $p=1$, and $s=3$, i.e., consisting of five masses with the vibration absorber being deployed at the mass $m_1$, and with the objective to fully suppress vibration of mass $m_3$.  
 The harmonic force \eqref{eq:disturbance-cosine-form} acting at mass $m_5$ is characterized by $\bar f_d=\SI{1}{N}$ and frequency $\omega=\SI{23.248}{s^{-1}} (\SI{3.7}{Hz})$.  
 The nominal parameters of the structure are assumed as: $m_i=\SI{1}{kg},\ i=1, \dots, 4$, $m_5=\SI{2}{kg}$, $m_a=\SI{0,5}{kg}$, $k_i=\SI{750}{Nm^{-1}},\ i=1, \dots, 5$, $k_a=\SI{700}{Nm^{-1}}$, $c_i=\SI{2}{Nsm^{-1}},\ i=1, \dots, 5$, $c_a=\SI{2}{Nsm^{-1}}$.

The optimization problem is considered in the form \eqref{mini:multiobject-minimisation-problem-transformed}--\eqref{eq:constraint2:Wa-maximum}. \textcolor{black}{Contrary to the experimental case study}, the parameters to be determined are the stiffness at all the links $k_i,\ i=1, \dots, 6$, and a complete DR parameter set $m_a,\ k_a,\ c_a$, while all the other parameters of the system are assumed to be fixed. 
%On the other hand, the springs with tuneable stiffness are widely used in practice, see e.g. \cite{buhr1997non}.
%, \cite{8818640}. 
%Analogously, the damping parameters can well be adjusted in practice. However, they are not considered among the decisive parameters here in order to reduce dimensionality of the problem. 
\color{black}
Thus, the vector of decisive parameters is given by $\theta=[m_a, c_a, k_a, k_1, k_2, k_3, k_4, k_5, k_6]\tran$.
% \setcounter{MaxMatrixCols}{20}
% \begin{equation}\label{case-study:vector-theta-masses}
%       \theta
%       =
%       \begin{bmatrix} 
%         m_a & c_a & k_a & k_1 & k_2 & k_3 & k_4 & k_5 & k_6 
%       \end{bmatrix}\tran.
% \end{equation}
% because they are easy to be tuned. This is not the case for the damping coefficients as they are mainly determined by the friction and damping of the springs, no physical dampers are considered in the set-up. It is not easy to modify the stiffness, as the springs are not easy replaceable. However, if needed, these coefficients can be included in the parameter set too.
In \cref{tab:masses-decision-variables}, we refer to $\theta^{\mathrm{lb}}$ and $\theta^{\mathrm{ub}}$ as vectors of lower and upper bounds for our decision variables, respectively. Additionally, we denote $\theta^{0}$ as a vector of \emph{nominal} parameters.
  For the scaling purposes within the objective function, we assume $W_{\nom}(\theta^{0}) = \max\{ W_{i, \max}(\theta^{0}):\ i=1, \dots, 6 \}=\SI{0.0607}{J}$, $P_{\nom}(\theta^{0}) = P_{\max}(\theta^{0})= \SI{0.0111}{W}$.
  % \begin{align}
  %     & W_{\nom}(\theta^{0}) = \max\{ W_{i, \max}(\theta^{0}):\ i=1, \dots, 6 \}=\SI{0.0607}{J},\\
  %     & P_{\nom}(\theta^{0}) = P_{\max}(\theta^{0})= \SI{0.0111}{W}.
  % \end{align}
  Additionally, we select $\gamma=0.5$. Matrices $\mathbf{A} = [\mathbf{I}, -\mathbf{I}]\tran$, $\mathbf{b} =[\theta^{\mathrm{ub}}, -\theta^{\mathrm{lb}}]\tran$ define the linear inequality constraints from \eqref{eq:constraint2:Ax<=b}.
  \color{black}
  %take form $\mathbf{A} = [\mathbf{I}, -\mathbf{I}]\tran$, $\mathbf{b} =[\theta^{\mathrm{ub}}, -\theta^{\mathrm{lb}}]\tran$
  % \begin{align}
  %     \mathbf{A} = \begin{bmatrix}
  %         \mathbf{I} \\ -\mathbf{I}
  %     \end{bmatrix}, \
  %     \mathbf{b} =
  %     \begin{bmatrix}
  %         \theta^{\mathrm{ub}} \\ -\theta^{\mathrm{lb}}
  %     \end{bmatrix}.
  % \end{align}
  %for nominal values of lower and upper bounds respectively we refer to Table \ref{tab:masses-decision-variables}.
Stability is guaranteed by imposing a boundary $\xi_{\alpha}=\SI{-0.2}{s^{-1}}$ in \eqref{eq:constraint2:abscissa-stable}, ensuring that the spectral abscissa of the optimized system will not be worse than the nominal one (we refer to \cref{tab:criterion}).
%ii) a good balance is achieved between stability margin and optimality.
To enforce inequality constraint \eqref{eq:constraint:Wa-maximum}, we require the maximum elastic potential energy at the absorber's link to be lower than $\xi_{a} = 0.01 \ \text{J}$, which can be considered as a rather conservative value.

\color{black}
For determining the spectral abscissa, \verb|tds_sa| function from the Matlab \emph{TDS-CONTROL} toolbox  \cite{appeltans2022tds}, \cite{appeltans2023analysis} is used (with setting:  \verb|rhp=-1.2|,  \verb|max_size_evp=1200|, \verb|newton_tol=1e-10| and \verb|newton_max_iter=20|). This function implements a spectral discretization method \cite{wu2012reliably} for determining position of the rightmost characteristic roots. The defined optimization task was solved using GRANSO (with  setting: \verb|mu0=5.5|, \verb|opt_tol=1e-8|, \verb|rel_tol=0|,  \verb|step_tol=1e-12|, \verb|viol_ineq_tol=0|). The required 
%Note also that the GRANSO requires determining the value of 
gradients of the objective and constraint functions were approximated by finite differences with spacing \verb|1e-8|.
% \begin{table}[ht]
%       %\centering
%       \begin{tabular}{p{3cm}p{1cm}l}
%            {\em Variable name} & {\em Value} & {\em TDS-control description} \\
%            \verb|rhp| & \verb|-1.2| & right half plane lower bound \\
%            \verb|max_size_evp| & 1200 & maximal size of the generalized eigenvalue problem \\
%            \verb|newton_tol| & \verb|1e-10| & tolerance on the norm of the residuals in Newton's method \\
%            \verb|newton_max_iter| & \verb|20| & maximal number of Newton method iterations for each root \\
%       \end{tabular}
%       \caption{Used settings for computing spectral abscissa via TDS-CONTROL matlab package.}
%       \label{tab:settings-tds-control}
% \end{table}
% \begin{table}[ht]
%       %\centering
%       \begin{tabular}{p{3cm}p{1cm}l}
%            {\em Variable name} & {\em Value} & {\em GRANSO description} \\
%            \verb|mu0| & \verb|5.5| &  initial value of penalty parameter \\
%            \verb|opt_tol| & \verb|1e-8| & optimal tolerance\\
%            \verb|rel_tol| & \verb|0| & relative tolerance\\
%            \verb|step_tol| & \verb|1e-12| & step tolerance\\
%            \verb|viol_ineq_tol| & \verb|0| & violation tolerance of the inequality constraints\\
%            \verb|maxit| & \verb|500| & maximum iterations \\
%       \end{tabular}
%       \caption{Used settings for GRANSO optimization.}
%       \label{tab:settings-granso}
% \end{table}
\begin{table}[t]
      \centering
      \begin{tabular}{cc|cccc}
        variable & units & lower bound & upper bound & nominal & optimized \\
        $\theta$ & & $\theta^{\mathrm{lb}}$ & $\theta^{\mathrm{ub}}$ & $\theta^{0}$ & $\theta^{*}$ \\
        \hline
           $m_a$ & $\SI{}{kg}$ & 0.2 & 2.0 & 0.5 & 0.675 \\
           $c_a$ & $\SI{}{Nsm^{-1}}$ & 1.0 & 10.0 & 2.0 & 4.134 \\
           $k_a$ & $\SI{}{Nm^{-1}}$ &400 & 2000 & 700 & 699.863 \\
           $k_1$ & $\SI{}{Nm^{-1}}$ &400 & 2000 & 750 & 736.119 \\
           $k_2$ & $\SI{}{Nm^{-1}}$ &400 & 2000 & 750 & 761.605 \\
           $k_3$ & $\SI{}{Nm^{-1}}$ &400 & 2000 & 750 & 770.249 \\
           $k_4$ & $\SI{}{Nm^{-1}}$ &400 & 2000 & 750 & 599.090 \\
           $k_5$ & $\SI{}{Nm^{-1}}$ &400 & 2000 & 750 & 727.512 \\
           $k_6$ & $\SI{}{Nm^{-1}}$ &400 & 2000 & 750 & 530.197 \\
      \end{tabular}
      \caption{Nominal and optimized values of decision variables $\theta$, together with lower and upper bounds.}
      \label{tab:masses-decision-variables}
\end{table}
\begin{table}[t]
     \centering
     \begin{tabular}{l|l|ll|ll|ll}
          & \multicolumn{3}{c|}{\small{Objective function} ($\gamma = 0.5$)} & \multicolumn{2}{c|}{\small{Inequality constraints}} & \multicolumn{2}{c}{\small{Control law}} \\
          \multicolumn{1}{c|}{setup} &  $J~\text{[-]}$ & $\max\limits_{i}W_{i,\max}~\text{[J]}$ & $P_{\max}~\text{[W]}$ & $\alpha~\text{[}\SI{}{s^{-1}}\text{]}$ &  $W_{a, \max}~\text{[J]}$ & $g~\text{[}\SI{}{Nm^{-1}}\text{]}$ &  $\tau~\text{[s]}$ \\ \hline
          nominal & \textbf{1.0} & 0.01115 &  0.06067 & -0.20926 & 0.00238& -129.96 & 0.04617 \\
          optimized & \textbf{0.111} & 0.00128 & 0.00653 & -0.22208 & 0.00005& -368.53 & 0.01682\\
     \end{tabular}
     \caption{Comparison of objective function values, inequality constraints and control law parameters for the nominal and optimized parameter settings.}
     \label{tab:criterion}
\end{table}

%\subsection{Solution of the optimization problem}
Optimization task was performed from 100 starting points $\theta$ selected randomly respecting the defined upper and lower bounds. 
%For multiple starting points, penalty function dominated the start of the optimization process as many of the starting points violated stability constraint \eqref{eq:constraint2:abscissa-stable}.
Optimization by GRANSO run under MATLAB R2022a on a workstation with AMD Ryzen~5 2600 (6 physical cores, 3.4 GHz) and 32 GB RAM. %Solving time for most of the starting points $\theta$ varied from 3 to 6 minutes. %mostly because of computation of the gradient of spectral abscissa. 
The overall solution took 535 minutes. 
% \color{black}
The optimized parameter set is given in \cref{tab:masses-decision-variables} whereas the values of the criterion and quantities bound by the  constrains are given in \cref{tab:criterion}. Notice that, interestingly, the decision parameter optimization lead to relatively small adjustments compared to their nominal values. 
%Concerning the stiffness at the links set nominally to $k_i=\SI{750}{Nm^{-1}},\ i=1, \dots, 6$, it varies only from $k_6=\SI{530.197}{Nm^{-1}}$ to $k_3=\SI{770.249}{Nm^{-1}}$, despite of relatively wide range allowed by the parameter constraints $k_i\in [400, 2000]\,\mathrm{Nm^{-1}},\ i= 1, \dots, 6$. Analogous results are obtained for the absorber, where the mass $m_a$ was adjusted from $\SI{0.5}{kg}$ to $\SI{0.675}{kg}$, the stiffness $k_a$ remained practically unchanged at $\SI{700}{Nm^{-1}}$. The largest relative change can be observed for the damping $c_a$, which was adjusted from $\SI{2}{Nm^{-1}}$ to $\SI{4.134}{Nm^{-1}}$.
%
The objective function $J$ was substantially reduced from the normalized value $1.0$ to $0.111$, projecting the reduction of the maximum of the elastic potential energy $W_{\max}$ from $\SI{0.01115}{J}$ to $\SI{0.00128}{J}$, and of DR actuator power $P_{\max}$ from $\SI{0.06067}{W}$ to $\SI{0.00653}{W}$.
%i.e. almost ten times in both the considered quantities. 
It can also be seen in \cref{tab:criterion} that the spectral abscissa was kept slightly below the $\alpha$ constraint set to $\xi_{\alpha}=\SI{-0.2}{s^{-1}}$. Furthermore, the safety constraint at the DR link $\xi_a=\SI{0.01}{J}$ was satisfied with very large margin as $W_{a,\max}$ was reduced from $\SI{0.00238}{J}$ to negligible $\SI{0.00005}{J}$. The DR feedback \eqref{eq:resonator-ut} parameters, were determined by 
%are given in \cref{tab:criterion}. Taking into account 
%\cref{proposition:DR-parameters}, they both were obtained by
\eqref{eq:DRpar2} with $k=0$.
%
%For both nominal and optimized setups, the DR feedback \eqref{eq:resonator-ut} parameters are given in \cref{tab:criterion}. Taking into account \cref{proposition:DR-parameters}, they both were obtained by \eqref{eq:DRpar2} with $k=0$. For the system with optimized parameters, the delay $\tau$ is almost three times smaller, while the magnitude of the gain $g$ is almost three times larger, compared to the system with nominal parameters.
 
% \subsection{Numerical results}

\color{black}
 Distributions of the spectra computed by \verb|tds_roots| function of \emph{TDS-CONTROL} toolbox (with the same setting used for \verb|tds_sa|) are shown in \cref{fig:spectrums} for both the nominal and optimized systems. Due to the involvement of the delay in the applied DR feedback, both systems represented by DDAE \eqref{eq:stability:DDAE} (though with different parameters) have infinite number of roots \textcolor{black}{with retarded distribution}. 
 %As can be seen in the left subfigure of \cref{fig:spectrums} with the large-scale ranges, for both nominal and optimized systems, we have a single retarded root chain %\cite{hale2003stability} 
 %located safely far to the left of the imaginary axis. 
 \textcolor{black}{From the detail of the rightmost spectrum,} 
 %shown in the right subfigure of \cref{fig:spectrums}. 
 \textcolor{black}{the fulfillment of the predefined spectral abscissa constraint can be checked visually.} %as all the roots lie to the left of the dashed line at $\xi_a=\SI{-0.2}{s^{-1}}$}. It can also be observed that the distribution of the rightmost roots was only slightly affected by the parameter optimization, while it led to a fair shifting of the retarded chain more to the left.    
\begin{figure}[t]
    \centering
    \includegraphics[width=0.9\textwidth]{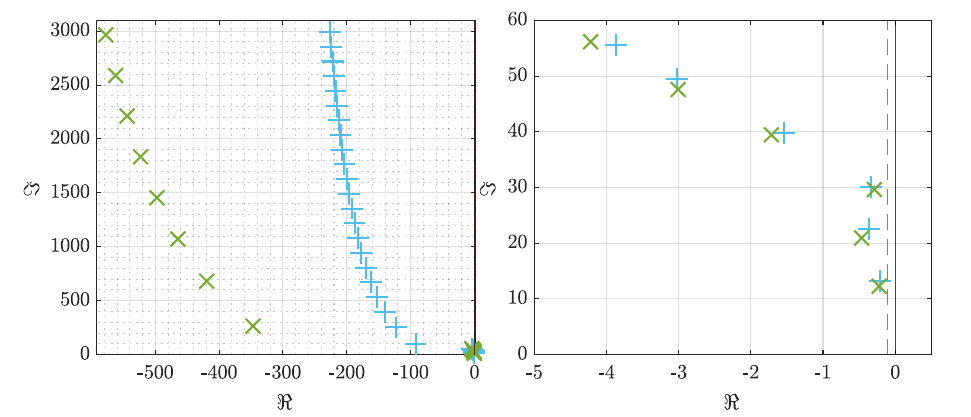}
    \caption{Spectrum distribution of the system with the DR, represented by DDAE \eqref{eq:stability:DDAE} with both nominal parameters (blue, +), and optimized parameters (green, x). {\em Left} --- Large scale view; {\em Right} --- Rightmost, stability determining roots with visualized value of the spectral constraint  $\xi_{\alpha}=-0.2$ (red, dashed).}
    \label{fig:spectrums}
\end{figure}
\begin{figure}[p]
    \centering
    \includegraphics{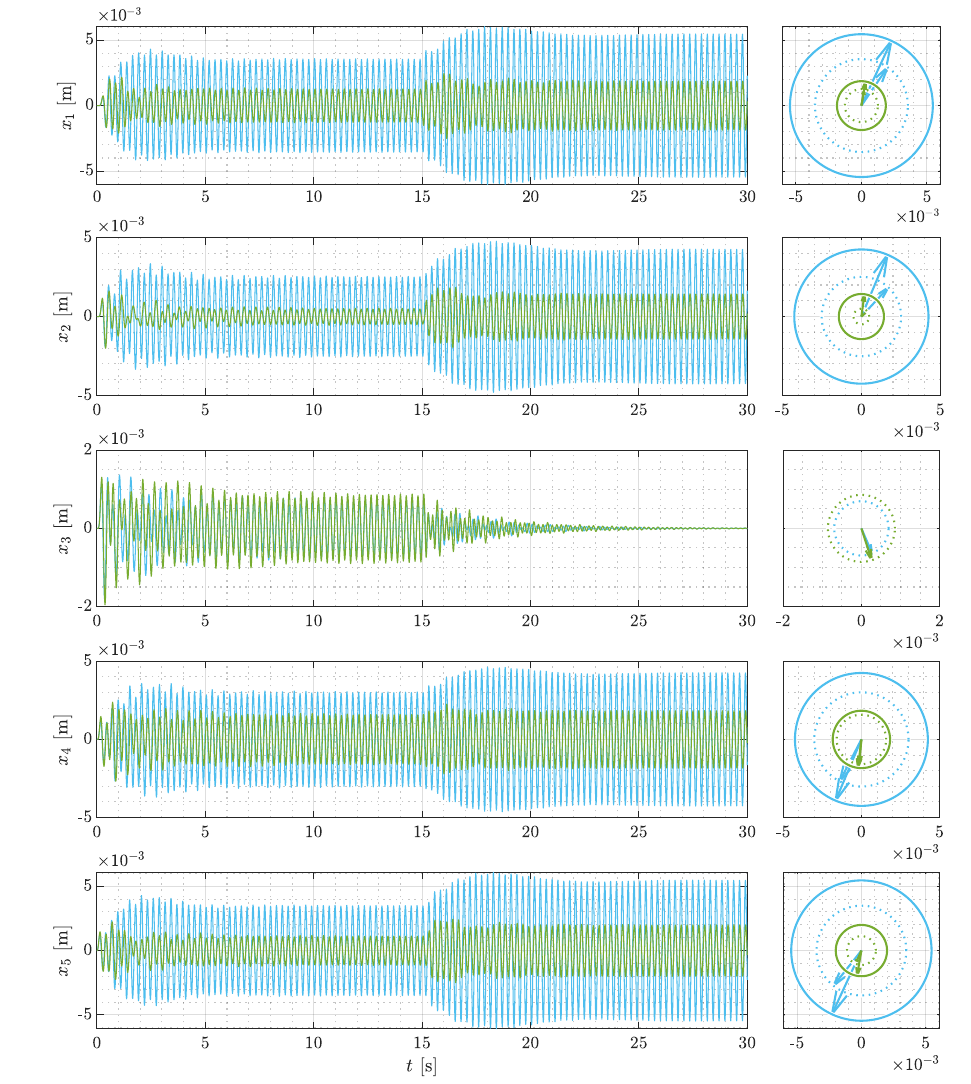}
    \caption{Motion of the masses $m_i,\ i=1, \dots, 5$ for the whole system described by \eqref{eq:stability:DDAE} with nominal parameters (blue) and optimized parameters (green). {\em Left} --- Simulated transients from zero initial conditions: passive regime at $t\in[0, 15)\, \mathrm{s}$; active regime $t\in[15, 30]\, \mathrm{s}$. {\em Right} --- Equilibrium oscillations in phasor view determined by \eqref{eq:passive-positions} for passive regime (dotted), and either by  \eqref{eq:prop1-proof-solved-xr} for $x_1, x_2$, or by \eqref{eq:prop1-proof-solved-xv} for $x_4, x_5$ for the active regime (solid).}
    \label{fig:positions}
\end{figure}
\begin{figure}[p]
    \centering
    \includegraphics{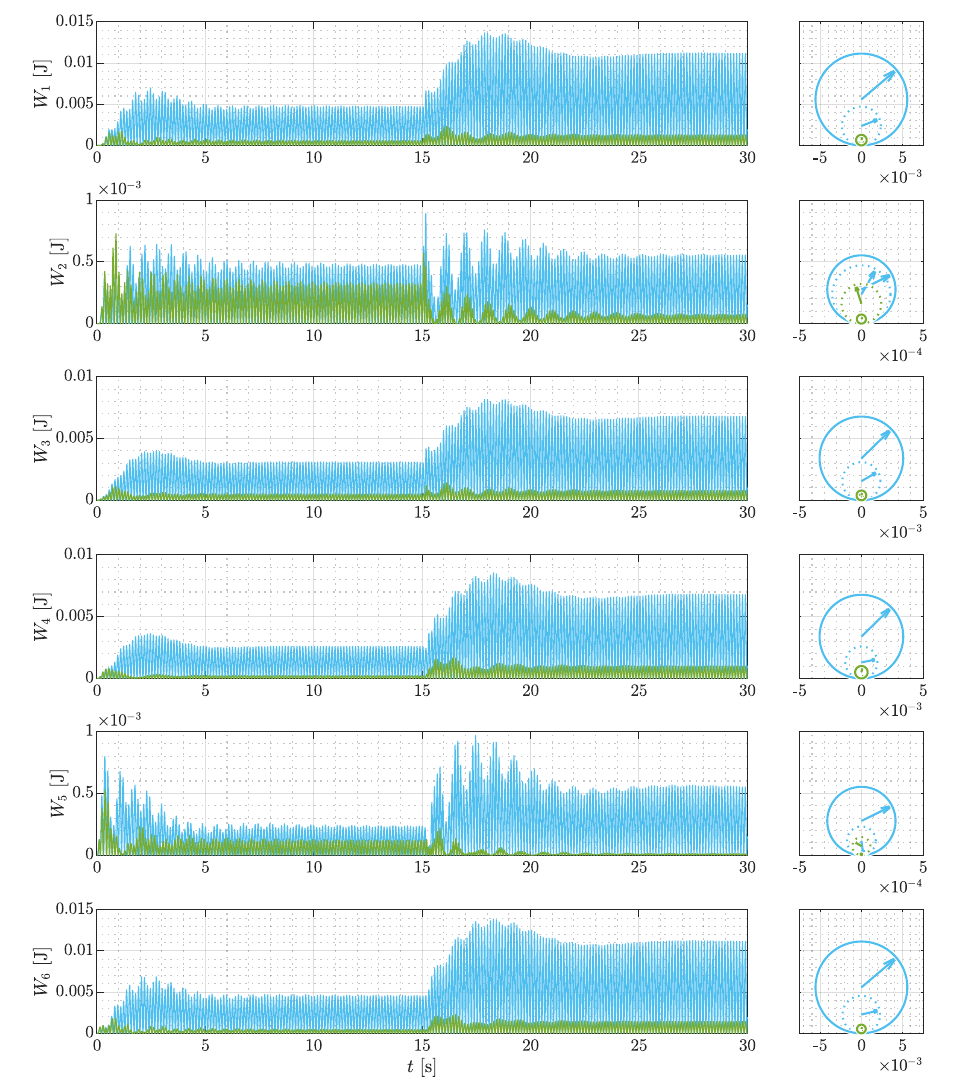}
    \caption{Potential elastic energies in links $i=1, \dots, 6$ for the whole system \eqref{eq:stability:DDAE} with nominal parameters (blue) and optimized parameters (green). {\em Left} --- Simulated transients from zero initial conditions: passive regime at $t\in[0, 15)\, \mathrm{s}$; active regime $t\in[15, 30]\, \mathrm{s}$. {\em Right} --- Equilibrium oscillations in phasor view determined by \eqref{eq:passive-elastic-energy-phasor-time} with the center point shifted by time average \eqref{eq:passive-elastic-energy-phasor-mean} for passive regime (dotted), and by \eqref{eq:elastic-energy-phasor-time} with the center point shifted by \eqref{eq:elastic-energy-phasor-mean} for the active regime (solid).}
    \label{fig:elastic-energies}
\end{figure}
\begin{figure}[p]
    \centering
    \includegraphics{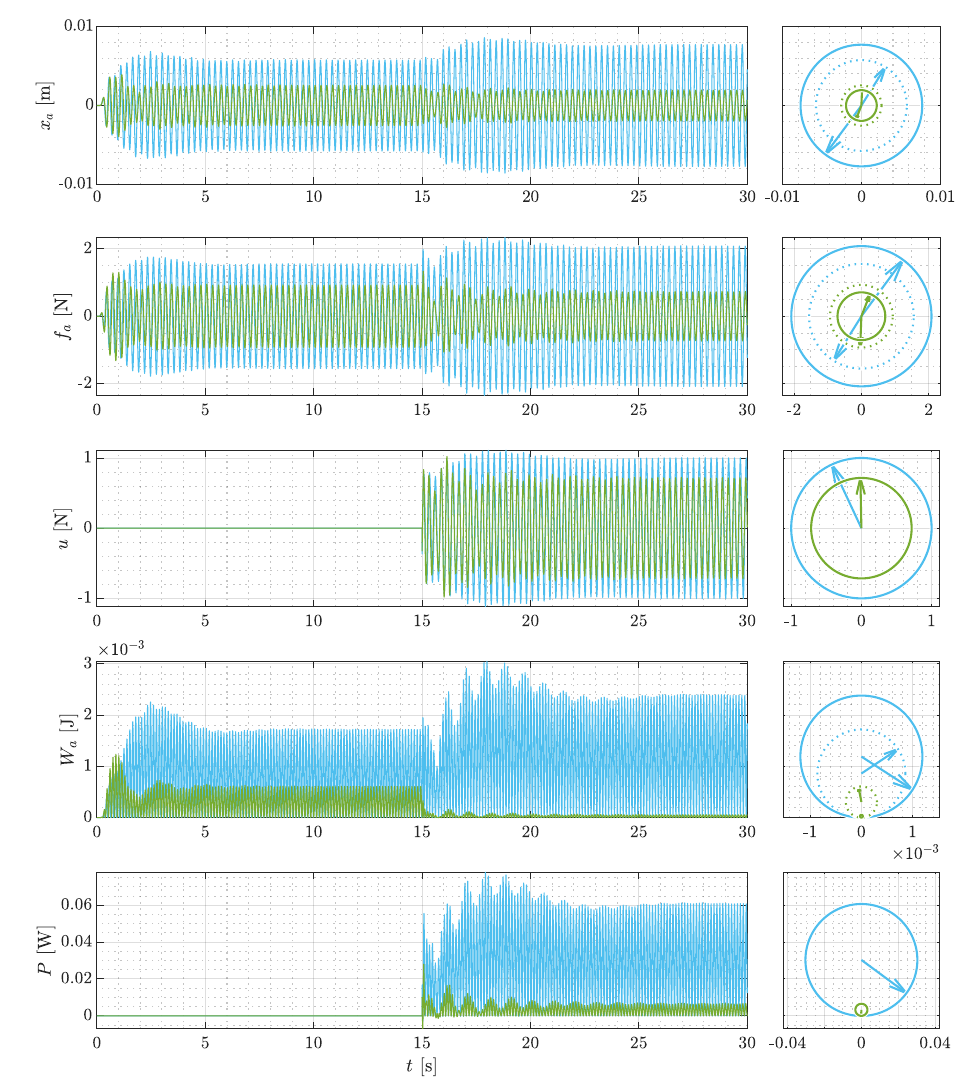}
    \caption{Absorber's characteristics: motion of mass $m_a$; total force acting on absorber $f_a$, control force $u$, elastic potential energy between $m_1$ and $m_a$ and actuator power $P$, with nominal parameters (blue) and optimized parameters (green). {\em Left} --- Simulated transients from zero initial conditions: passive regime at $t\in[0, 15)\, \mathrm{s}$; active regime $t\in[15, 30]\, \mathrm{s}$.  {\em Right} --- Equilibrium oscillations \textcolor{black}{of according quantities in phasor view for both passive regime (dotted, see \ref{secap}) and active regime (solid, see Section \ref{sec:dr-anaylisis-design}).}}
    %position of $m_a$ by \eqref{eq:passive-position-ma} for passive regime (dotted) and by \eqref{eq:resonator-vector-xa} for active regime (solid); force $f_a$ by \eqref{eq:fa-passive2} for passive regime (dotted) and by \eqref{eq:prop1-fa-normal-case} for active regime (solid); control force $u$ by \eqref{eq:resonator-ut-function-fat} for active regime (solid); elastic potential energy \eqref{eq:passive-elastic-energy-phasor-time-absorber} shifted by time average \eqref{eq:passive-elastic-energy-phasor-mean-absorber} for passive regime (dotted) and \eqref{eq:elastic-energy-phasor-time-absorber} shifted by \eqref{eq:elastic-energy-phasor-mean-absorber} for active regime (solid); actuator power by \eqref{eq:resonator-ut-phasor} with the center point shifted by time average \eqref{eq:resonator-ut-phasor-time-average} for active regime.}
    \label{fig:absorber}
\end{figure}
% \begin{figure}[t]
%     \centering
%     \includegraphics{./elastic_energies_bar.pdf}
%     \caption{Comparison of nominal (blue) and optimized (green) setup of maxima of elastic potential energies in passive regime (left) and active regime (right).}
%     \label{fig:elastic-energies-bar}
% \end{figure}
%\begin{figure}[t]
%    \centering
%    \includegraphics{./sa_frequency_sweep.pdf}
%    \caption{Spectral abscissa \eqref{eq:stability:spectral-abscissa} with respect to the excitation frequency $\omega$ for the system with nominal (blue, thick) and optimized (green, thick) parameters. Additionally, spectral abscissa for optimized setup for additional delay branches $k$ is shown (green -- $g^{-}$, yellow -- $g^{+}$).}
%    \label{fig:discussion-spectral-abscissa}
%\end{figure}
 
 \color{black}
 Results of simulations performed in Matlab-Simulink (solver \verb|ode45|,  relative tolerance \verb|1e-6|) for both nominal and optimized systems are shown in the left subfigures of \cref{fig:positions}--\cref{fig:absorber}. 
 %The simulations were performed in Simulink, \verb|ode45| {\em solver} with {\em relative tolerance} \verb|1e-6|.
\color{black}
 Starting from zero initial conditions, first, the transients of the passive regime are shown in the time range $t\in[0, 15] \, \mathrm{s}$. 
 %It takes (practically) $\SI{10}{s}$ until the steady-state oscillations are reached. 
 Subsequently, at $t=\SI{15}{s}$, the DR's control law \eqref{eq:resonator-ut} is turned on, which results in transients heading towards complete vibration suppression of the target mass $m_3$, which is (practically) done for $t>\SI{25}{s}$. In the right sub-figures of \cref{fig:positions}--\cref{fig:absorber}, the amplitudes and relative position of the phasors for the steady-state oscillations are shown, which were obtained analytically for both active and passive regimes (for the passive, see \ref{secap}). 
 %In what follows, we pay particular attention to the comparison of either amplitudes or maxima of the quantities at the steady-state oscillations.     

In \cref{fig:positions}, we can see that turning the DR feedback on results in increasing equilibrium motion amplitudes for $x_1$, $x_2$, $x_4$, and $x_5$ compared to the passive regime. Though, the amplitudes at the active regime for the optimized system are smaller compared to the amplitudes at the passive regime for the nominal system.
In \cref{fig:elastic-energies}, the mass motion is projected to elastic potential energies at the connecting links $W_i,\ i=1, \dots, 6$. In the active regime, the maxima of these energies in all the links of the optimized system are substantially reduced compared to the nominal system. Notice also that for the optimized system, when turning the control on, the maxima of energies decrease in the links 2 and 5, while they increase for the links 1, 3, 4 and 6. On the other hand, for the nominal system, the maxima of all the energies increase when comparing passive and active regimes. %In a uniform scale, the maxima of energies are also compared in \cref{fig:elastic-energies-bar}. 
\textcolor{black}{Note also that $W_{\max}$ from \cref{tab:criterion} corresponds to the link no. 1.}
 
 Subsequently, in \cref{fig:absorber}, all the observed quantities for the DR absorber are shown. Interestingly, for the optimized system, not only the equilibrium amplitude of position $x_a$, but also the amplitude of the overall force $f_a$ in the absorber's link are slightly smaller for the active regime compared to the passive regime. Both passive and active amplitudes of the optimized system are considerably smaller than the corresponding amplitudes of the nominal system. As regards the control action amplitude, it is only slightly smaller for the optimized system compared to  the nominal system. Relatively large difference in favor of the optimized system can be seen in the maxima of targeted elastic potential energies $W_a$, already for the passive regime. In the active regime, the differences are substantially amplified. Analogous substantial reduction can be observed for the maximum of the needed power, see also \cref{tab:criterion} with the numerical values of $W_{a,\max}, P_{\max}$.

\color{black}

\section{Conclusions}
\label{section:conclusions}
%\color{red}
A thorough analysis of non-collocated vibration absorption of a system composed of a series of flexibly linked masses and excited by a harmonic disturbance force is performed. Next to the motion amplitudes, the elastic potential energies in the flexible links between the masses are studied with the aim to assess the risk of fatigue across the structure. 
%The analysis is based on replacing active absorber by an equivalent force acting on the structure.
It is shown that as soon as vibrations at the target mass are fully suppressed,  the system can be separated into a vibrating subsystem excited by the disturbance force and the resonating subsystem excited by the force generated by the active vibration absorber. At this stage, the motion amplitudes and consequently the maxima of the elastic potential energies in the links are uniquely determined by the actuation force amplitude, its frequency and the parameters of the system, and they cannot be changed neither by the absorber parameter adjustment nor by its control feedback. Thus, if the risk of fatigue is to be minimized across the structure, it needs to be done by optimizing the structural parameters of the system. 

A preliminary step towards the subsequent optimization based design is a complete parameter assessment of the delayed resonator with position feedback and the stability analysis in terms of determining the spectral abscissa of the whole system. 
%described by a set of Delay Differential Algebraic Equations. 
Attention is also paid to determining the power needed for keeping the target mass stopped at the steady-state oscillation stage and to differences between the maxima of potential energies at the passive and active regimes. 
Taking the findings of the structural analysis into account, the simultaneous structural optimization and control design of the delayed resonator are translated into an optimization problem. Its objective function is selected to balance the maximum from the potential energy set, corresponding to the link most prone to fatigue, and the power needed by the delayed resonator to absorb the vibration at the target. The objective function to be minimized is supplemented by a set of constraints, including stability of the overall system, construction constraints, and safety constraint on the absorber's link. 
%The objective function as well as some of the constraints are in general nonconvex and nonsmooth. 

\textcolor{black}{For validation of the analytical results, two case studies are included. In the first one performed for the experimental setup with three linked masses, the mechanisms by which the optimized parameters influence the optimization objective is explained and experimentally validated. In the numerical case study, the proposed optimization based design is applied to a system consisting of five linked masses}. It is demonstrated that the problem can be well solved by the GRANSO optimization tool. The case study documents that the resistance of the system to fatigue as well as the needed power can be substantially reduced by the proposed approach. 
Further research directions include, the use of more complex feedback laws to either increase stabilization potential or suppress more than one frequency, and the extension to more general structures, \textcolor{black}{studying and compensating the effect of nonlinearities with subsequent experimental validation on a system with more carts}. This will also include the involvement of acceleration feedback which will turn the system spectrum posture from retarded to neutral.
\color{black}

\section*{Acknowledgements}
The authors thank Professor Nejat Olgac, University of Connecticut, for motivating this work and for a valuable discussion concerning the presented research results.

\appendix
\section{Characteristics of the passive regime}\label{secap}
For consistency, we outline here selected formulas for the steady state of the passive case, mirroring those derived above for the steady-state oscillations of the non-collocated vibration suppression. Assuming $\vec{u}=0$, and the relation between $\vec{x}_a(t)$ and $\vec{x}_p(t)$ from \eqref{eq:resonator-vector-absorber}, the force in the absorber's link, in the phasor form, is given by
% \begin{equation}\label{eq:fa-passive}
% \vec{f}_a(t)=(k_a+\mathrm{j}\omega c_a)(\vec{x}_p-\vec{x}_a).
% \end{equation}
%By the relation between $\vec{x}_a$ and $\vec{x}_p$ from \eqref{eq:resonator-vector-absorber}, the expression \eqref{eq:fa-passive} can be rewritten as
\begin{equation}\label{eq:fa-passive2}
 \vec{f}_a = \frac{-m_a\omega^2(\jmath\omega c_a + k_a)}{k_a-m_a\omega^2+\jmath\omega c_a}\vec{x}_p.
\end{equation}
Using $\vec{x}_p = \mathbf{e}_p\tran \vec{\mathbf{x}}$, the system phasor balance by \eqref{eq:sys-general-vector-form} can be turned to
\begin{equation}\label{eq:sys-general-vector-form-passive}
    \mathbf{P}(\omega) \vec{\mathbf{x}} = \mathbf{B}_d\vec{f}_d,
\end{equation}
where
\begin{equation}
      \mathbf{P}(\omega) = \mathbf{A}(\omega) +\mathbf{B}_a\mathbf{e}_p\tran \frac{m_a\omega^2(\jmath\omega c_a+k_a)}{k_a-m_a\omega^2+\jmath\omega c_a}.
\end{equation}
This allows us to express the positions of masses in a linear chain as
\begin{equation}\label{eq:passive-positions}
    \vec{x}^\mathrm{p}_i = \mathbf{e}_i\tran \mathbf{P}^{-1}(\omega) \mathbf{B}_d\vec{f}_d,
\end{equation}
and position of absorber $m_a$ as
\begin{equation}\label{eq:passive-position-ma}
      \vec{x}^\mathrm{p}_a = \frac{\jmath\omega c_a + k_a}{k_a - m_a\omega^2 + \jmath\omega c_a} \mathbf{e}_p\tran \mathbf{P}^
{-1}(\omega) \mathbf{B}_d\vec{f}_d.
  \end{equation}
  the elastic potential energies are defined as
  \begin{equation}\label{eq:passive-elastic-energies-vector-form}
     W^\mathrm{p}_i(t) = \Bar{W}^\mathrm{p}_{i, \mean} + \Re\left\{ \vec{W}^\mathrm{p}_i e^{\jmath 2\omega t}\right\},
  \end{equation}
where, assuming the difference between positions $\Delta \vec{x}^\mathrm{p}_{i, i-1} = \vec{x}^\mathrm{p}_i - \vec{x}^\mathrm{p}_{i-1}$
\begin{equation}\label{eq:passive-elastic-energy-phasor-mean}
     \Bar{W}^\mathrm{p}_{i, \mean} = \frac{1}{4}k_i \Delta \vec{x}^\mathrm{p}_{i, i-1} (\Delta \vec{x}^\mathrm{p}_{i, i-1})^{*},
\end{equation}
and
\begin{equation}\label{eq:passive-elastic-energy-phasor-time}
     \vec{W}^\mathrm{p}_{i} = \frac{1}{4}k_i (\Delta\vec{x}^\mathrm{p}_{i, i-1})^2.
\end{equation}
The maximal value of \eqref{eq:passive-elastic-energies-vector-form} is then given by
\begin{equation}
    W_{i, \max}^\mathrm{p} = \frac{1}{2} k_i \left| \left( (\mathbf{e}_i\tran - \mathbf{e}_{i-1}\tran)\mathbf{P}^{-1}(\omega)\mathbf{B}_d\vec{f}_d\right)^2\right|.
\end{equation}
The potential energies in the links by the base are handled by considering $\mathbf{e}_{0} = \mathbf{e}_{d+1} = \mathbf{o}_{d}$.
Finally, the elastic potential energy in the absorber's link is given by
\begin{equation}\label{eq:passive-elastic-energies-vector-form-absorber}
    W^\mathrm{p}_a(t) = \Bar{W}^\mathrm{p}_{p, \mean} + \Re\left\{ \vec{W}^\mathrm{p}_a e^{\jmath 2\omega t} \right\},
\end{equation}
where, assuming the difference between positions $\Delta \vec{x}^\mathrm{p}_{a, p} = \vec{x}^\mathrm{p}_a - \vec{x}^\mathrm{p}_{p}$,
\begin{equation}\label{eq:passive-elastic-energy-phasor-mean-absorber}
    \Bar{W}^\mathrm{p}_{a, \mean} = \frac{1}{4}k_a (\Delta \vec{x}^\mathrm{p}_{a, p})^{*} \Delta \vec{x}^\mathrm{p}_{a, p}
\end{equation}
and
\begin{equation}\label{eq:passive-elastic-energy-phasor-time-absorber}
     \vec{W}^\mathrm{p}_{a} = \frac{1}{4}k_a (\Delta\vec{x}^\mathrm{p}_{a, p})^2.
\end{equation}
Maximum of \eqref{eq:passive-elastic-energies-vector-form-absorber} can be expressed using \eqref{eq:passive-positions} and \eqref{eq:passive-position-ma} as 
\begin{equation}
    W_{a,\max}^\mathrm{p} = \frac{1}{2} k_a \left| \left(1-\frac{\jmath\omega c_a + k_a}{k_a - m_a\omega^2 + \jmath\omega c_a}\right) \mathbf{e}_p\tran \mathbf{P}^{-1}(\omega) \mathbf{B}_d\vec{f}_d \right|^2.
\end{equation}

\bibliographystyle{elsarticle-num}
%\bibliographystyle{elsarticle-harv}
%\bibliographystyle{ksfh_nat}
%\bibliography{non-collocated,references,mainT,mainTold}
\bibliography{references}

\end{document}